\documentclass[a4paper]{article}

\usepackage{amsmath}
\usepackage{amsthm}
\usepackage{amssymb}
\usepackage{graphics}
\usepackage{tikz}

\theoremstyle{plain}
\newtheorem{theorem}{Theorem}
\newtheorem{proposition}[theorem]{Proposition}
\newtheorem{corol}[theorem]{Corollary}
\newtheorem{lemma}[theorem]{Lemma}

\theoremstyle{definition}
\newtheorem{example}[theorem]{Example}
\newtheorem{definition}[theorem]{Definition}

\newcommand{\U}{U} 

\newcommand{\V}{V} 
\newcommand{\lab}{\Lambda} 
\newcommand{\R}{R} 
\newcommand{\G}{\mathcal{G}} 
\newcommand{\HG}{\mathcal{H}} 
\newcommand{\A}{\mathcal{A}} 
\newcommand{\B}{\mathcal{B}} 

\newcommand{\mG}{{\overline{\G}}} 
\newcommand{\C}{\mathcal{C}} 
\newcommand{\F}{\mathcal{F}} 
\newcommand{\degree}{\mathrm{deg}} 
\newcommand{\Zet}{Z}

\newcommand{\W}{\underline W}

\newcommand{\cempty}{0} 
\newcommand{\id}{1'} 
\newcommand{\di}{0'} 
\newcommand{\all}{1} 
\newcommand{\conv}[1]{#1^{-1}}
\newcommand{\compl}[1]{#1^c}
\newcommand{\cpi}{\bar{\pi}} 
\newcommand{\compos}{\circ} 
\DeclareMathOperator{\rres}{\backslash} 
\DeclareMathOperator{\lres}{/} 

\newcommand{\bisimeq}[2]{\simeq^{#1}_{#2}} 
\newcommand{\mysim}[2]{\preceq^{#1}_{#2}} 
\newcommand{\indistk}[2]{\equiv^{#1}_{#2}} 
\newcommand{\contndk}[2]{\Rrightarrow^{#1}_{#2}} 
\newcommand{\indist}[1]{\equiv^{#1}} 
\newcommand{\contnd}[1]{\Rrightarrow^{#1}} 
\newcommand{\BiSim}[3]{\mathrm{BiSim}^{#1}(#2,#3)}
\newcommand{\Simforth}[3]{\mathrm{Sim}_{\rm forth}^{#1}(#2,#3)}
\newcommand{\Simback}[3]{\mathrm{Sim}_{\rm back}^{#1}(#2,#3)}

\newcommand{\typek}[3]{\mathrm{tp}^{#1}_{#2}(#3)}
\newcommand{\type}[2]{\mathrm{tp}^{#1}(#2)}
\newcommand{\atype}[2]{\mathrm{atp}^{#1}(#2)}
\newcommand{\atypeq}[1]{\mathrm{atp}^{#1}}
\DeclareMathOperator{\aexp}{aexp}

\DeclareMathOperator{\paths}{paths}
\newcommand{\faths}[1]{\paths^\F_{#1}}
\newcommand{\fpaths}{\paths^\F}
\DeclareMathOperator{\upaths}{upaths}
\newcommand{\graph}[1]{\mathit{graph}(#1)}
\newcommand{\ugraph}[1]{\mathit{ugraph}(#1)}

\newcommand{\edgelabel}[1]{\ensuremath{\mathit{#1}}}

\newcommand{\knows}{\edgelabel{knows}}
\newcommand{\worksAt}{\edgelabel{worksAt}}
\newcommand{\patientOf}{\edgelabel{patientOf}}
\newcommand{\hasDisease}{\edgelabel{hasDisease}}

\newcommand{\Atom}{\mathrm{Atom}}

\newcommand{\phigenp}{\varphi^{\F,k}_{\mG_1,P}}
\newcommand{\phiposatoms}{\varphi^{\F}_{\mG_1,\text{posatoms}}}
\newcommand{\phinegatoms}{\varphi^{\F}_{\mG_1,\text{negatoms}}}
\newcommand{\phicompforth}{\varphi^{\F,k}_{\mG_1,\text{composition forth}}}
\newcommand{\phicompback}{\varphi^{\F,k}_{\mG_1,\text{composition back}}}
\newcommand{\phiprojforth}{\varphi^{\F,k}_{\mG_1,\text{projection forth}}}
\newcommand{\phiprojback}{\varphi^{\F,k}_{\mG_1,\text{projection back}}}
\newcommand{\philresforth}{\varphi^{\F,k}_{\mG_1,\text{leftres forth}}}
\newcommand{\philresback}{\varphi^{\F,k}_{\mG_1,\text{leftres back}}}
\newcommand{\phirresforth}{\varphi^{\F,k}_{\mG_1,\text{rightres forth}}}
\newcommand{\phirresback}{\varphi^{\F,k}_{\mG_1,\text{rightres back}}}
\newcommand{\psigenp}{\psi^{\F,k}_{\mG_1,P}}

\newcommand{\psicompforth}{\psi^{\F,k}_{\mG_1,\text{composition forth}}}
\newcommand{\psicompback}{\psi^{\F,k}_{\mG_1,\text{composition back}}}
\newcommand{\psiprojforth}{\psi^{\F,k}_{\mG_1,\text{projection forth}}}
\newcommand{\psiprojback}{\psi^{\F,k}_{\mG_1,\text{projection back}}}
\newcommand{\psicprojforth}{\psi^{\F,k}_{\mG_1,\text{coproj forth}}}
\newcommand{\psicprojback}{\psi^{\F,k}_{\mG_1,\text{coproj back}}}
\newcommand{\psilresforth}{\psi^{\F,k}_{\mG_1,\text{leftres forth}}}
\newcommand{\psilresback}{\psi^{\F,k}_{\mG_1,\text{leftres back}}}
\newcommand{\psirresforth}{\psi^{\F,k}_{\mG_1,\text{rightres forth}}}
\newcommand{\psirresback}{\psi^{\F,k}_{\mG_1,\text{rightres back}}}

\begin{document}

\begin{titlepage}

\title{Similarity and bisimilarity notions appropriate for
characterizing indistinguishability in fragments of the calculus of relations} 
\author{George H.L. Fletcher \\ Eindhoven University of
Technology \and Marc Gyssens \\ Hasselt University \\
transnational University of Limburg \and Dirk Leinders \\ Hasselt
University \\ transnational university of Limburg
\and Jan Van den Bussche\thanks{Corresponding author.  Full
address: Jan Van den Bussche, Hasselt University, Martelarenlaan 42, 3500 Hasselt,
Belgium. Email: jan.vandenbussche@uhasselt.be. Phone: +32 11 26
82 26.} \\ Hasselt University \\ transnational
University of Limburg \and Dirk Van Gucht \\ Indiana University
\and Stijn Vansummeren \\ Universit\'e Libre de Bruxelles \\ (ULB)}
\maketitle

\thispagestyle{empty}

\begin{abstract}
Motivated by applications in databases,
this paper considers various fragments of the calculus of binary relations.
The fragments are obtained by leaving out, or keeping in,
some of the standard operators, along with some derived operators
such as set difference, projection, coprojection, and
residuation.  For each considered fragment, a characterization is
obtained for when two given binary relational structures are
indistinguishable by expressions in that fragment.  The
characterizations are based on appropriately adapted notions of simulation
and bisimulation.

\bigskip \textbf{Keywords:} calculus of relations,
indistinguishability, bisimulation, simulation, coprojection,
residuation

\end{abstract}

\end{titlepage}

\section{Introduction} \label{sec:introduction}

The calculus of relations \cite{Tarski41,Givant06,Maddux91,
Pratt92} consists of five natural operations on binary
relations: union, intersection, complementation, composition, and
converse.  These operators can be applied to given binary
relations, combined with the four standard constant relations:
empty, all, identity, and diversity.  The calculus of relations
is a very natural formalism and occurs within logics for
reasoning about binary relations, notably dynamic and description
logics \cite{dynamiclogicbook,dlhandbook}.  The calculus also has
historically motivated the development of the theory of relation algebras
\cite{maddux_book,hh_relalgames}.  In the present paper, however,
we are not looking at abstract relation algebras, but rather at
the question of \emph{indistinguishability} of two given finite
binary relational structures within the calculus of relations.

This paper has been inspired by the authors' ongoing research
program to understand in a precise way the expressive power of
the calculus of relations as a database query language for binary
relation structures
\cite{good_tarski,gsvg_tagging,marc_pods2006,wu_pospathfragment,rafragments,amai_dipi}.
Indistinguishability of structures in various logics is one of
the most basic tools in the study of the expressive power of
database query languages as well as in finite model theory
\cite{EF99,Libkin04,ahv_book}.  Indeed the calculus of relations, as
a core relational algebra query language on binary relations, is
very relevant to the field of databases.  Binary relations
(directed graphs)
show up naturally in data on the
Web \cite{ABS99,FLM98}, dataspaces \cite{FHM05}, Linked Data
\cite{BHB09}, and RDF data \cite{RDFPrimer}.  Moreover, in
restriction to directed graphs that are trees, the relational
calculus is closely tied to the XML query language XPath, and the
expressive power of XPath and various fragments has been
intensively investigated \cite{BFK05,Marx05,MR05,marc_pods2006}.

Here, working with general finite binary relation structures
rather than trees, we consider, in addition to the five binary
relation operations and four constant binary relations mentioned
above, also four derived operations that are well known in the
literature: set difference; projection;
coprojection; and residuation.  These derived operations can be
expressed in terms of the other operations and constants, but can
still be interesting on their own when considering fragments
where some other operations or constants have been left out.  We
consider set difference because it is the standard
domain-independent alternative to complementation in database
query languages \cite{ahv_book}.  We consider projection and coprojection
(existential and universal quantification) because they are
standard logical operations, and have been shown important in the
XPath setting \cite{MR05}, so it is natural to study their behaviour when
generalising from trees to general graphs.  Finally, we consider
residuation because it is similar to the standard relational
division operation in databases, and corresponds to the set
containment join \cite{mamoulis_setjoin}.  Obviously, one could
keep on inventing additional operations on binary relations and
study their interdependencies, but our chosen set of
operations is not too large and well-motivated from the point of
view of query languages.

Our goal now is to understand the relative importance of the
various operations and the effect of their presence on
indistinguishability.  Thereto we consider all possible fragments
of the calculus of relations that can be constructed as follows.
The most basic fragment we consider has the empty and identity
relations as constants, and the operations union, composition,
and intersection.  Then all other fragments arise by adding any
choice of the remaining operations and constants.  For each
fragment, we provide a characterization of when two finite binary
relation structures are indistinguishable by expressions in the
fragment.  Our approach follows the one outlined by Goranko and
Otto \cite{gorankotto}: we provide new notions of finite-round
(degree-bounded) bisimulations, appropriate for fragments of the
calculus of relations, and give characteristic expressions for
them.  For finite structures, such an approach immediately leads
to a Hennessy--Milner-type theorem
\cite{hennessymilner,blackburn_modallogic}.

One may ask why intersection is present in all the fragments we
consider.  Intersection is the most basic query language
primitive \cite{ahv_book}. Our results rely heavily on its
presence, in the same way as the classical Hennessy--Milner
theorem relies heavily on the presence of conjunction in the
modal logic.  Intersection is known not to be ``safe for
bisimulation'' \cite{benthem_safe,blackburn_modallogic}.
Nevertheless, the question whether or not all operations in a
certain query language are safe for bisimulation appears to be
quite different from the problem of the present paper: that of
determining whether or not two given structures can be
distinguished in a certain query language.

This is not to say, however, that indistinguishability of
structures in fragments lacking intersection is uninteresting.
But it changes the nature of the problem so drastically that we
leave omitting intersection outside the scope of the present
paper.  For example, consider the fragment consisting only of
composition and nothing else.  Then indistinguishability of
finite structures amounts to the equivalence problem for finite
automata, which is PSPACE-complete \cite{ahu}.  In contrast, for
all fragments considered in this paper, we will see that
indistinguishability is decidable in polynomial time.

As mentioned above, bisimilarity-like characterizations of
indistinguishability are common in modal logics.  For non-modal
logics, such as first-order logic, indistinguishability is
typically captured by Ehrenfeucht-Fra\"\i ss\'e games.  There
also exist intermediate fragments of first-order logic, such as
the guarded and the packed fragments
\cite{marxvenema_fmta}, where indistinguishability
can still be captured by appropriate notions of bisimulation.
Note, however, that most of the fragments considered in this
paper are not subsumed by the packed fragment.  For example, the
expressions $(R
/ S) \circ T$, or $((R \circ S) - (R \circ T)) \circ S$, or
$\cpi_1(R) \circ S \circ \cpi_2(T)$, are not expressible
in the packed fragment.\footnote{In
the example, $R/S$ stands for the left residual, $-$ stands for
set difference, and $\cpi$ stands for coprojection; these
operators will be defined in the next Section.}  Note als that
projection such as as occurring in, e.g., $\pi_1(R /S)$, is in
general not equivalent to guarded existential quantification.
Only the ``positive'' fragments, that include none of the
residuals, set difference, complementation, and coprojection,
fall in the packed fragment.

To conclude this Introduction, we note another motivation to
understand indistinguishability in database query language
fragments, apart from the relevance to expressive power and
the intrinsic foundational motivation.
This is the new approach of \emph{structural} indexing to
database query processing, proposed by some of us and others
\cite{FVW09,wu_pospathfragment,indexingtriples},
whereby a given query expression is
processed by accessing blocks of data indistinguishable by the
operations used in the given expression.
By the results of our work, these blocks can be computed using
similarity or bisimilarity checks.

\paragraph{Summary}
The further contents of this paper may be summarized as follows.
In Section~2 we define the language fragments formally, and
define the notion of indistinguishability.  In Section~3 we
discuss different ways how indistinguishability can be
characterized; in particular we discuss the connection with
multi-dimensional modal logics, and the 3-variable fragment of
first-order logic.  In Section~4, we define finite-round
bisimulations appropriate for the fragments with the set
difference operation.  In Section~5, we define finite-round
simulations appropriate for the fragments without set difference.
In Section~6 we given Hennessy--Milner-type theorems for
indistinguishability of finite structures.  We conclude in
Section~7.

\section{Language fragments and indistinguishability}
\label{sec:lang-struct-indist}

We assume an infinite universe of atomic data elements, denoted by
$\U$.  A \emph{binary relation} on $\U$ is a subset of $\U^2=U
\times U$.  We further fix an arbitrary finite set $\lab$ of
\emph{relation names}, called the \emph{vocabulary}.
In the calculus of relations, a
\emph{structure} is a pair $\G = (\V, (R^{\G})_{R \in
\Lambda})$ where $\V$ is a subset of $\U$ and each $R^{\G}$ is a
binary relation on $\V$. The set $\V$ is called the set of nodes
of $\G$; the vocabulary $\Lambda$ can be thought of as a set of
edge labels whereby $\G$ can be thought of as an edge-labeled
directed graph.  When $\V$ is finite, the structure is said to be
a \emph{finite} structure.

Expressions in the calculus of relations are built recursively from
the relation names $R \in \Lambda$,
and the constant symbols empty ($\cempty$), all
($\all$), diversity ($\di$), and identity ($\id$), using
the following standard and/or derived operations. The standard
operations are union $(e_1 \cup e_2)$, intersection $(e_1 \cap e_2)$,
complementation $(\compl{e})$, composition $(e_1 \compos e_2)$, and
converse $(\conv{e})$; the derived operations we consider are set
difference $(e_1 - e_2)$, projection ($\pi_1{e}$ or $\pi_2{e}$),
co-projection ($\cpi_1{e}$ or $\cpi_2{e}$), left residual ($e_1\lres
e_2$) and right residual ($e_1\rres e_2$).\footnote{To distinguish
  between set difference and the right residual, we use the minus sign
  ($-$) for set difference.}

Semantically, on any structure $\G$ as above, an expression $e$ defines a
binary relation, denoted by $e(\G)$. For convenience, we recall the semantics of
the constants and the standard operations.
\begin{align*}
  \R(\G)             &= \R^{\G};\\
  \cempty(\G)                &= \emptyset;\\
  \all(\G)                        &= \V^2;\\
  \di(\G)                         &= \{(s,t) \mid s,t \in \V\ \&\ s \neq t\};\\
  \id(\G)                         &= \{(s,s) \mid s \in \V\};\\
  (e_1 \cup e_2) (\G)         &= e_1(\G) \cup e_2(\G);\\
  (e_1 \cap e_2) (\G)         &= e_1(\G) \cap e_2(\G);\\ 
  \compl{e}(\G)              &= \{(s,t) \mid s,t\in V\
  \&\ (s,t) \notin e(\G)\};\\
  (e_1 \compos e_2)(\G)   &= \{(s,t) \mid (\exists v)((s,v)\in e_1(\G)\ \& \ (v,t)\in e_2(\G))\};\\
  (\conv{e})(\G)                 &= \{(s,t) \mid (t,s) \in e(\G)\}.
\end{align*}
The semantics of the derived operations is as follows:
\begin{align*}
  (e_1 - e_2)(\G)               &= \{(s,t) \mid (s,t) \in e_1(\G)
  \ \&\ (s,t) \notin e_2(\G)\}\\
  \pi_1(e)(\G)                 &= \{(s,s) \mid (\exists t)(s,t) \in e(\G)\}\\
  \pi_2(e)(\G)                 &= \{(s,s) \mid (\exists t)(t,s) \in e(\G)\}\\
  \cpi_1(e)(\G)               &= \{(s,s) \mid s\in V \ \& \ \lnot(\exists t)(s,t) \in e(\G)\}\\
  \cpi_2(e)(\G)               &= \{(s,s) \mid s\in V \ \& \ \lnot(\exists t)(t,s) \in e(\G)\}\\
  (e_1\lres e_2)(\G)          &= \{(s,t) \mid (\forall v)((t,v) \in e_2(\G) \rightarrow (s,v) \in e_1(\G)) \}\\
  (e_1\rres e_2)(\G)          &= \{(s,t) \mid (\forall v)((v,s) \in e_1(\G) \rightarrow (v,t) \in e_2(\G)) \}
\end{align*}

\begin{figure}
\begin{center}
\resizebox{0.9\textwidth}{!}{\includegraphics{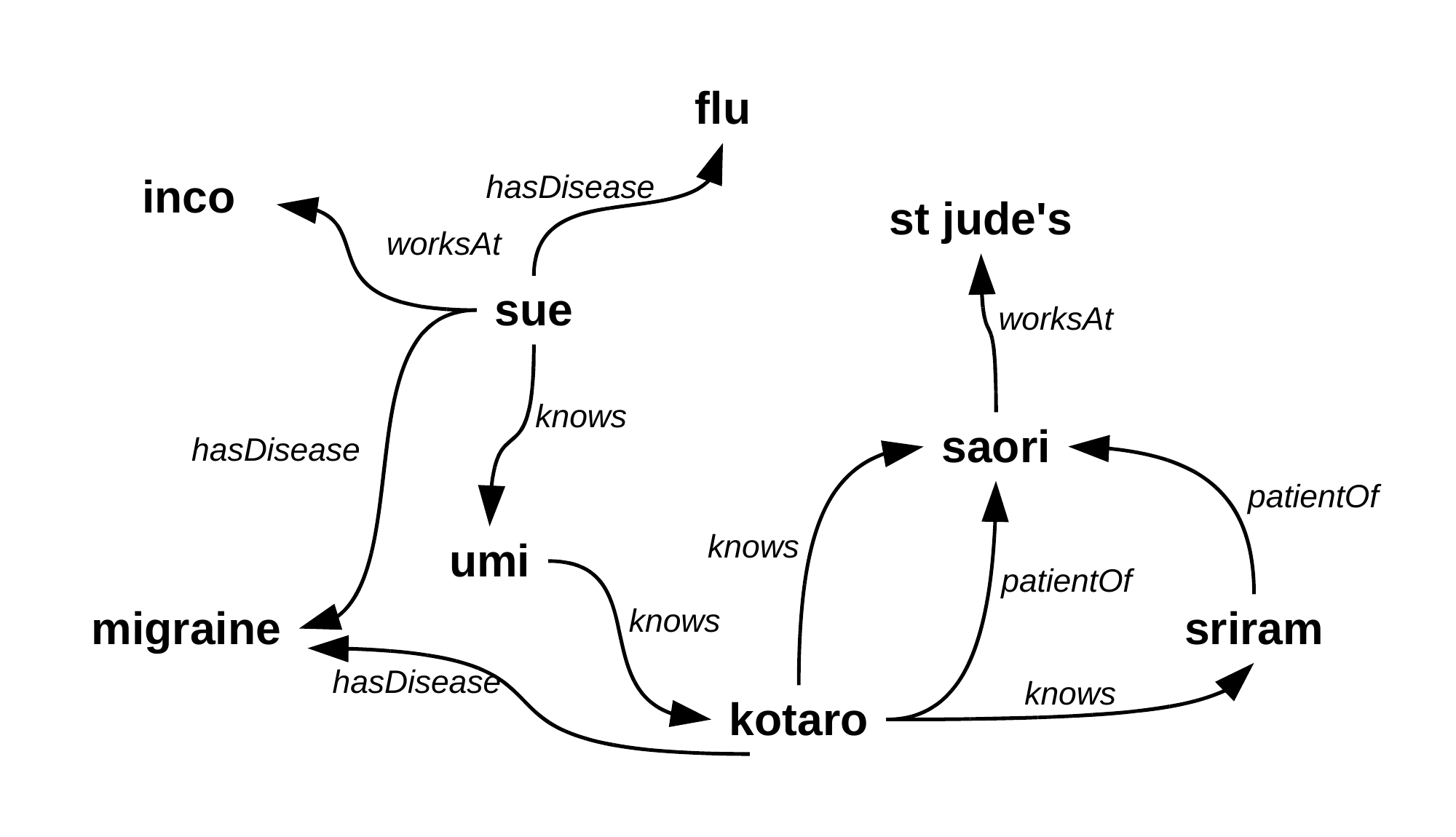}}
\end{center}
  \caption{Example structure from Example~\ref{ex:intro-example}.}
  \label{fig:intro-example}
\end{figure}

\begin{example} \label{ex:intro-example}
Figure~\ref{fig:intro-example} shows a finite structure $\G$.
The set of nodes equals
$\{\text{migraine}, {\rm flu}, {\rm sue}, {\rm umi}, {\rm saori},
{\rm sriram}, \text{st jude's}, {\rm inco}\}$,
and the vocabulary $\Lambda$ equals $\{\knows, \worksAt,
  \patientOf, \hasDisease \}$.
\begin{itemize}
\item
  The doctors (i.e., persons having patients), can be retrieved
  from $\G$ by the expression
  \begin{equation*}
  \label{eq:doctors}
  e_1 = \pi_2(\patientOf)
\end{equation*}
resulting in $e_1 ( \G ) = \{ \text{(saori,saori)} \}$.
\item
 The people and the doctors they know can be obtained by the expression
 \begin{equation*}
   \label{eq:know-doctor}
   e_2 = \knows \circ e_1
 \end{equation*}
 resulting in $e_2(\G) = \{ \text{(kotaro,saori)} \}$.
 \item
 The doctors and the hospitals where they practice:
 \begin{equation*}
   \label{eq:doctor-hospital}
   (e_1 \circ \worksAt) (\G) = \{ \text{(saori,st jude's)} \}.
 \end{equation*}
 \item
Ill people without medical care:
 \begin{equation*}
   (\pi_1(\hasDisease) - \pi_1(\patientOf)) (\G) = \{ \text{(sue,sue)} \}.
 \end{equation*}
 \item
 Healthy doctors:
 $$ (e_1 \cap \cpi_1(\hasDisease))(\G) = \{ \text{(saori,saori)} \}. $$
 \item
 Finally, the doctors who know all the patients of some other
 doctor can be retrieved by the expression
 $$ e_1 \cap \pi_1((\knows \lres \patientOf) \cap \di), $$ which
 on our example graph yields the empty relation, since the graph
 contains only one doctor.
\end{itemize}
\end{example}

\subsection{Queries and equivalence}
Expressions in the calculus of relations express queries.
Formally, a \emph{query} is a mapping $Q$ from the set of all
structures to the set of all binary relations on $U$, such that
for each structure $\G$, if $V$ is the node set of $\G$, then
$Q(\G)$ is a binary relation on $V$.  Obviously the query $Q$
expressed by an expression $e$ is simply defined by $Q(\G) :=
e(\G)$.

Two expressions $e_1$ and $e_2$ are now called \emph{equivalent}, denoted by
$e_1 \equiv e_2$, if they express the same query, i.e., if
$e_1(\G) = e_2(\G)$ for all possible
structures $\G$.  The following equivalences demonstrate that the
derived operations are indeed derived, and also present some
additional interdependencies among the constants and operations
considered in this paper:
\begin{align*}
& 1 \equiv 0^c \equiv \id \cup \di  \equiv 0 \lres 0 \equiv 0 \rres 0\\
& 0' \equiv {1'}^c \\
& e_1 - e_2 \equiv e_1 \cap e_2^c \\
& e^c \equiv 1 - e \\
& \pi_1(e) \equiv
(e \compos \conv e) \cap \id \equiv (e \compos 1) \cap \id
\equiv \cpi_1(\cpi_1(e)) \\
& \pi_2(e) \equiv
(\conv e \compos e) \cap \id \equiv (1 \compos e) \cap \id
\equiv \cpi_2(\cpi_2(e)) \\
& \cpi_i(e) \equiv 1' - \pi_i(e) \\
& e_1 \lres e_2 \equiv \compl{(\compl e_1 \compos \conv e_2)}
\\
& e_1 \rres e_2 \equiv \compl{(\conv e_1 \compos \compl e_2)}
\end{align*}

Of course the above list of equivalences is by no means complete.
For example, another well-known equivalence is $(e_1 \circ
e_2)^{-1} \equiv e_2^{-1} \circ e_1^{-1}$.  It will be useful to
have the following generalization of this equivalence:

\begin{proposition} \label{convprop}
Every expression $e$ is equivalent to an expression $e'$ that
uses the same operations as $e$, and in which converse
is only applied to relation names.
\end{proposition}
\begin{proof}
We actually show the claim not for $e$ but for $e^{-1}$; the
claim for $e$ then follows by applying it to each topmost
application of converse within $e$.
We go by induction on the structure of $e$.
The case where $e$ is a relation name is trivial.
The constants are all equivalent to
their converse.  When $e$ is of the form $e_1^{-1}$, we
have $e^{-1} \equiv e_1$, which
can be put in the required form by the induction hypothesis.
When $e$ is $e_1 \cup e_2$, we have $\conv e \equiv \conv{e_1}
\cup \conv{e_2}$, and similarly when $e$ is $e_1 \cap e_2$ or $e_1 - e_2$.
When $e$ is $e_1 \circ e_2$, we have $\conv e \equiv \conv{e_2}
\circ \conv{e_1}$.  When $e$ is $\pi_1(e_1)$, we
have $\conv e \equiv e$, and similarly for 
$\pi_2(e_1)$, $\cpi_1(e_1)$, and $\cpi_2(e_2)$.  Finally, we have
$\conv{(e_1 \lres e_2)} \equiv \conv{e_2} \rres \conv{e_1}$ and
$\conv{(e_1 \rres e_2)} \equiv \conv{e_2} \lres \conv{e_1}$.
\end{proof}

\subsection{Language fragments} We will consider various fragments
of the calculus of relations.  The most basic fragment we
consider is denoted by $\C$: it has the constants $\cempty$ and
$\id$ and the operators composition, union, and intersection. All
other fragments are defined by adding to $\C$ some additional
constants and operators.

Formally, for any subset\footnote{In order to
simplify our presentation somewhat, we only consider fragments
containing both the first and second projection ($\pi_1$ and
$\pi_2$) or none of them, and similarly for coprojection.  This
simplification is not essential to our results, however.}
$\F$ of $\{\di,\all, {}^{-1},{}^c,
\pi,\cpi, -,\lres,\rres\}$, we define the fragment $\C(\F)$
consisting of the expressions built up from the relation names,
$0$, $\id$, and the constants from $\F$,
using the operations composition, union,
intersection, and the operations from $\F$.

The fragment $\C({}^{-1},{}^c)$ already amounts to the full calculus,
since, by the equivalences listed above,
all other operations can be derived in it.
More precisely, from the listed equivalences, we can note the following:
\begin{itemize}
\item
Any fragment containing complement also includes $\all$, $\di$,
difference, projection, and coprojection.
\item
Any fragment containing $\di$ also includes $\all$.
\item
Any fragment containing converse, $\di$, or coprojection,
also includes projection.
\item
Any fragment containing projection and difference also contains
coprojection.
\item
Any fragment containing both converse and complement also
contains both residuals.
\end{itemize}
Accordingly, we say about a fragment $\C(\F)$ that
\begin{itemize}
\item
\emph{1 is present in $\C(\F)$ at degree 0} if $\F$ contains 1, $0'$, or
complement;
\item
\emph{1 is present in $\C(\F)$ at degree 1} if 1 is not 
present at degree 0, and $\F$ contains the left or right residual.
\item
\emph{1 is absent from $\C(\F)$} if 1 is neither present at
degree 0 nor at degree 1.
\end{itemize}
The idea behind these notions is that $1$ is expressible as
$0 \lres 0$ or as $0 \rres 0$.  We will see later that these two
expressions have degree one, as opposed to the expressions
for 1 that have degree zero, viz., 1 itself, $0' \cup 1'$,
or $0^c$.  The distinction between presence at degree 0, presence at
degree 1, and absence of 1 in a fragment will
manifest itself in Definition~\ref{defpath}.

That $1$ is expressible in terms of $0$ and the residuals shows that
our choice in this paper to include $0$ by default in all fragments is not
totally innocent, at least not in the presence of the residuals.
We do not anticipate, however, that
adapting the results of this paper to a setting where $0$ is
absent should require new techniques.  On a related note,
it is also interesting to point out\footnote{We thank the
anonymous referee for this remark.} that using $0$ and the left
residual, the first coprojection becomes expressible as $\cpi_1(e)
= (0 / e) \cap 1'$.

\subsection{Degrees and paths}
It is customary in finite model theory \cite{EF99}
to parameterize characterizations of indistinguishability
by the quantifier rank of formulas.  In our setting, the
role of quantifier rank will be played by what we call the degree.

For an expression $e$, we define the
\emph{degree} $\degree(e)$ of $e$ as follows.
Every relation name and constant symbol has degree zero.  Then,
\begin{align*}
  &\degree(e_1 \cup e_2) = \degree(e_1\cap e_2) = \degree(e_1 - e_2) = \max(\degree(e_1),\degree(e_2));\\
  &\degree(\compl{e}) = \degree(\conv{e}) = \degree(e);\\
  &\degree(e_1 \compos e_2) = \degree(e_1 \lres e_2) = \degree(e_1 \rres e_2) = 1 + \max(\degree(e_1), \degree(e_2));\\
  &\degree(\pi_1(e)) = \degree(\pi_2(e)) = \degree(\cpi_1(e)) = \degree(\cpi_2(e)) = 1 + \degree(e).
\end{align*}
The degree of an expression is the maximum depth of nested
applications of the composition, projection, co-projection, and the
left and right residual operation. Intuitively, the degree corresponds
to the quantifier rank of $e$ translated into first-order logic.

For a fragment $\F$ of the calculus of relations and a natural
number $k$, we denote the set of expressions in $\C(\F)$ of degree at
most $k$ by $\C(\F)_k$.

Before introducing the crucial Definition~\ref{defpath} below,
we need to agree on a natural way to view structures as directed, or
as undirected graphs.

\begin{definition} \label{defgraph}
Let $\G = (\V, (R^{\G})_{R \in \Lambda})$ be a structure.
Then $\graph \G$ is defined as the directed graph $(V,E)$ where
$E$ equals the set of all pairs $(x,y)$ in $V^2$ such that $(x,y)
\in R^{\G}$ for some $R \in \Lambda$.  Moreover, $\ugraph G$ is
defined as the undirected version of $\graph G$, i.e., as the
undirected graph $(V,E')$ where $E'$ is the set
of all unordered pairs $\{x,y\}$ such that $(x,y)$ or $(y,x)$ 
belongs to $E$.

For any natural number $k$, we further define $\paths_k(\G)$ as
the set of all pairs $(x,y)$ in $V^2$ such that there is a path
from $x$ to $y$ in $\graph \G$ of length at most $2^k$.  (The length of
a path equals its number of edges, and we agree that there is
a path of length $0$ from $x$ to $x$ for any $x \in V$, i.e.,
$(x,x)$ is always in $\paths_k(\G)$ for any $k$.)  
We define $\upaths_k(\G)$ similarly, but considering paths in the
undirected graph $\ugraph \G$.
\end{definition}

We now give:

\begin{definition}[$\F$-$k$-path] \label{defpath}
  Let $\C(\F)$ be a fragment of the calculus of relations and
  let $\G$ be a structure.  For any natural number $k$, we define
$\paths_k^\F(\G)$ as follows.
\begin{itemize}
\item
First, consider the case where 1 is absent in $\C(\F)$.
If $\F$ does not contain converse, then
$\paths_k^\F(\G)$ is defined to be $\paths_k(\G)$; if $\F$ does
contain converse, then $\paths_k^\F(\G)$ is defined to be
$\upaths_k(\G)$.
\item
Next, assume 1 is present at degree 1.
Then $\paths_0^\F(\G)$ is defined exactly as above for $k=0$, but
$\paths_k^\F(\G)$ for $k>0$ is simply $1(\G)$.
\item
Finally, 1 is present at degree 0, then 
$\paths_k^\F(\G)$ is again simply $1(\G)$ for all $k$ including
zero.
\end{itemize}
\end{definition}

We immediately note the following

\begin{lemma} \label{lempaths}
For each $\F$ and each $k$, the query $\paths_k^\F$ is
expressible in $\C(\F)_k$.
\end{lemma}

In the (easy) proof of the Lemma, we will use the notion of
\emph{atomic expression} which will also be used throughout the paper, so
we define it separately here:

\begin{definition}[Atomic expressions $\aexp(\F)$]
\label{defatom}
The \emph{atomic} expressions are those from
the finite set $\Atom = \{\id,\di\} \cup
\{R, \conv R \mid R \in \lab \})$.
The set of atomic expressions belonging to $\C(\F)$ is denoted by
$\aexp(\F)$.
\end{definition}
Note that by Proposition~\ref{convprop}, we can indeed assume for
any expression that the leaves of its syntax tree are labeled by
atomic expressions.

Now to the 
\begin{proof}[Proof of Lemma~\ref{lempaths}]
When 1 is absent, the lemma follows from the following equivalences:
  \begin{align*}
    \paths_0^\F & \equiv \bigcup_{e \in \aexp(\F)} e\\
    \paths_{k+1}^\F & \equiv \paths_k^\F \cup (\paths^\F_{k} \circ
  \paths^\F_{k}).
  \end{align*}

When 1 is present at degree 1, $\paths_0^\F$ is expressed as above
and $\paths_k^\F$ being equivalent to 1, 
can be expressed by $0 \lres 0$ or $0 \rres 0$.
When 1 is present at degree 0, $\paths_k^F$ is again equivalent to 1
and expressible by 1 itself, $0' \cup 1'$, or $0^c$.
\end{proof}

The next Proposition shows the relevance of $\paths_k^\F$.

\begin{proposition} \label{leminpaths}
Let $\F$ be a fragment of the calculus of relations and let $\G$
be a structure.  For any natural number $k$ and any expression $e \in
\C(\F)_k$, we have $e(\G) \subseteq \paths_k^\F(\G)$.
\end{proposition}
\begin{proof}
By the definition of $\paths_k^\F$,
the statement of the lemma is trivial when 1 is present in
$\C(\F)$ at degree 0.
Also when 1 is present at degree 1, the lemma
is trivial, except
for the case $k=0$, but then $\paths_0^\F = \paths_0^{\F'}$ where
$\F'$ is obtained from $\F$ by removing the residuals.  Note that
1 is absent in $\F'$ and that $\C(\F')_0 = \C(\F)_0$.

Hence, it suffices to prove the lemma for the case that 1 is
absent from $\C(\F)$.  This means that $\F$ does \emph{not} contain 1,
$0'$, complement, and the residuals.

We now proceed by structural induction on $e$.  If $e \in
\aexp(\F)$, then $e(\G) \subseteq \paths_0^\F(\G)$ by definition of 
$\paths_0^\F(\G)$.

If $e$ is $e_1 \cup e_2$, $e_1 \cap e_2$, or $e_1 - e_2$, the
  result follows immediately from the induction hypothesis.

If $e$ is $\pi_1(e_1)$, $\pi_2(e_1)$, $\cpi_1(e_1)$, or
$\cpi_2(e_1)$, the result is immediate because $\pi_1(e_1)(\G)
\subseteq \id(\G) \subseteq \paths_k^\F(\G)$. (Similarly for
$\pi_2(e_1)$, $\cpi_1(e_1)$, and $\cpi_2(e_1)$.)

Finally, if $e$ is $e_1 \compos e_2$, let $k_1 = \deg(e_1)$, $k_2
= \deg(e_2)$, and $\ell = \max(k_1,k_2)$. Note that $k = \ell
+1$. Now assume $(a,b) \in e_1 \compos e_2 (\G)$. Then, for some $c
\in \V$, we have $(a,c) \in e_1(\G)$ and $(c,b) \in e_2(\G)$. By
induction, we have $(a,c) \in \paths_{k_1}^\F(\G) \subseteq
\paths^\F_\ell$ and $(c,b) \in
\paths_{k_2}^\F(\G) \subseteq \paths^\F_\ell$, whence
$(a,b) \in \paths_k^\F(\G)$ as desired.
\end{proof}

\subsection{Indistinguishability}
A \emph{marked structure} $\mG$ is a pair $(\G,a,b)$ where $\G$
is a relational structure, and $(a,b)$ is an ordered pair of
nodes from $\G$.  Let $\C(\F)$ be a fragment of the calculus of
relations, and let $k$ be a natural number.  The
\emph{$\C(\F)_k$-type} of $\mG$, denoted by $\typek \F k\mG$,
is defined as the set of all expressions $e \in \C(\F)_k$
such that $(a,b) \in e(\G)$.
For two marked structures $\mG_1 = (\G_1,a_1,b_1)$ and $\mG_2 = (\G_2,
a_2,b_2)$, we write $\mG_1 \contndk \F k \mG_2$ if
$\typek \F k{\G_1,a_1,b_1}
\subseteq \typek \F k{\G_2,a_2,b_2}$, i.e., for every
expression $e \in \C(\F)_k$ such that $(a_1,b_1) \in e(\G_1)$, also
$(a_2,b_2) \in e(\G_2)$.  We then say that $\mG_2$ is \emph{one-sided
indistinguishable} from $\mG_1$ in $\C(\F)_k$.
When both $\mG_1 \contndk \F k \mG_2$ and $\mG_2
\contndk \F k \mG_1$,
we say that $\mG_1$ and $\mG_2$ are \emph{indistinguishable in
  $\C(\F)_k$} and denote this by $\mG_1 \indistk \F k \mG_2$.

Recalling Definition~\ref{defatom}, we also define
the \emph{atomic} $\F$-type of $\mG$, denoted by $\atype \F \mG$,
as $\Atom \cap \typek \F 0\mG$.  Note that $\atype \F \mG$ is
always a subset of $\aexp(\F)$.

Since indistinguishability is the same as one-sided
indistinguishability in both directions, it is more general to
look for a characterization of one-sided indistinguishability,
and that is what we will do.  On the other hand,
when the fragment contains complement or difference,
one-sided indistinguishability actually coincides with
indistinguishability, except in a trivial case:

\begin{proposition} \label{vier}
Let $\C(\F)$ be a fragment of the calculus of relations so that
$\F$ contains complement or difference.  Let $\mG_1 =
(\G_1,a_1,b_1)$ and $\mG_2 = (\G_2, a_2,b_2)$ be two marked
structures, and let $k$ be a natural number.  Then $$ \mG_1
\contndk \F k \mG_2 \quad \Leftrightarrow \quad \mG_1
\indistk \F k \mG_2 $$ whenever $(a_1,b_1) \in
\paths_k^\F(\G_1)$.  When $(a_1,b_1) \notin \paths_k^\F(\G_1)$,
the one-sided indistinguishability $ \mG_1 \contndk \F k \mG_2 $
holds trivially, and $\mG_1 \indistk \F k \mG_2$ holds if and
only if $(a_2,b_2) \notin \paths_k^\F(\G_2)$.
\end{proposition}
\begin{proof}
We first show that, when $(a_1,b_1) \in \paths_k^\F(\G_1)$, then
$\mG_1 \contndk \F k \mG_2$ implies $\mG_2 \contndk \F k
\mG_1$.  Thereto, let $e \in \C(\F)_k$ such that $(a_2,b_2) \in
e(\G_2)$; we must show that $(a_1,b_1)$ belongs to $e(\G_1)$.
Assume, for the sake of contradition, that it does not.  Then
$(a_1,b_1) \in (\paths_k^\F - e)(\G_1)$.  Note that when $\F$
would not contain difference, it would contain complement, and
then the expression $\paths_k^\F - e$ can be equivalently
written in $\C(\F)$ as $e^c(\G_1)$.  In either case, the
expression has degree $k$, so, since $\mG_1 \contndk \F k
\mG_2$, we have $(a_2,b_2) \in (\paths_k^\F - e)(\G_2)$.  In
particular it follows $(a_2,b_2) \notin e(\G_2)$ which yields the
desired contradiction.

When $(a_1,b_1) \notin \paths_k^\F(\G_1)$, then
$\mG_1 \contndk \F k \mG_2$ is indeed voidly satisfied, since in that
case, by Proposition~\ref{leminpaths}, the $\typek \F k{\mG_1}$ is empty.
Moreover then, clearly
$\mG_1 \indistk \F k \mG_2$ iff $(a_2,b_2)$ does not
belong to $e(\G_1)$, for any $e \in \C(\F)_k$, either.  We now note that the
latter holds iff $(a_2,b_2) \notin \paths^\F_k(\G_2)$.  Indeed, the
only-if is clear since $\paths^\F_k$ belongs to $\C(\F)_k$;  the
if-direction is again given by Proposition~\ref{leminpaths}.
\end{proof}

Similarly to indistinguishability in $\C(\F)_k$, i.e., for a fixed
degree $k$, we are also interested in indistinguishability in an
entire fragment $\C(\F)$.  Thus define the
\emph{$\F$-type} of a marked structure $\mG$, denoted by
$\type \F \mG$, as the set of all expressions $e$
from $\C(\F)$ such that $(a,b) \in e(\G)$.  Using this notion of
type we can now define the indistinguishability notions $ \mG_1
\contnd \F \mG_2$ and $\mG_1 \indist \F \mG_2 $ similarly
to the fixed-degree case.

\section{Approaches to bisimilarity}
\label{secapproaches}
\newcommand{\ffull}{\F_{\rm full}}

Before discussing indistinguishability for fragments of the
calculus of relations, let us first look at the full calculus
$\C(\ffull)$, with $\ffull$ consisting of complement and
converse.  Tarski and Givant showed that the calculus has equal
expressive power as FO$^3(2)$: the formulas with two free variables
in the three-variable fragment FO$^3$ of first-order logic
\cite{TG87}.  For FO$^3$, we have the three-pebble
Ehrenfeucht-Fra\"{\i}ss\'{e} game as a
characterization~\cite{EF99,Libkin04}. Marx and Venema, however,
showed that FO$^3(2)$ has also the same expressive power as arrow
logic~\cite{MV97}, a branch of multi-dimensional modal logic
devised to provide a formalization for simple reasoning about
objects that are thought of as arrows. By this correspondence,
bisimulations in terms of back-and-forth conditions that are well
known from modal logic can be used to characterize fragments of
FO$^3$, and, hence, of the calculus of relations.

Concretely, the language of arrow logic is a modal language with the
dyadic operator $\circ$, the monadic operator $\otimes$, and the modal
constant $\mathrm{id}$. Formulas in arrow logic are built up from a
set of propositional variables and the modal constant $\mathrm{id}$,
using the operators $\circ$ and $\otimes$, and the boolean connectives
$\land$, $\lor$, $\neg$. Using propositional variables to denote edge
labels; by interpreting the modal constant $\mathrm{id}$ as being true
for pairs $(a,a)$ of identical nodes; by interpreting the monadic
operator $\otimes$ as being true for pairs $((b,a),(a,b))$ of
``arrows'' such that the first arrow is the converse of the second
arrow; and finally, by interpreting the dyadic operator $\circ$ as
being true for triples $((a,b),(a,c),(c,b))$ of arrows such that the
first one is obtained by composing the second and the third arrow, we
can apply the characterization theorem of modal logic to immediately
obtain a characterization for the full calculus of relations.
We will next make this more precise.

The notion of bisimulation for multi-dimensional
modal logic, specialized to the above interpretation of arrow
logic, becomes the following:
\begin{definition} \label{def:bisim-arrow}
  Let $\G_1$ and $\G_2$ be two structures with node sets $\V_1$ and
  $\V_2$, respectively. A non-empty relation $Z \subseteq \V_1^2
  \times \V_2^2$ is an \emph{arrow-logic
  bisimulation} between $\G_1$ and $\G_2$ if
  it satisfies the following conditions:\footnote{The attentive
  reader will notice that the converse-forth condition and the
  converse-back condition are identical.  This is a consequence
  of the symmetry of the converse operator.  We could have
  simplified the definition by removing one of the identical
  conditions, but preferred to stay in line with the
  general format of bisimulation conditions for multidimensional
  modal logic.}
  \begin{description}
  \item[Atoms] if $(a_1,b_1,a_2,b_2)$ is in $Z$, then $(a_1,b_1) \in
    R(\G_1)$ if and only if $(a_2,b_2) \in R(\G_2)$, for all
    $R \in \Lambda$; 
  \item[Forth] if $(a_1,b_1,a_2,b_2) \in Z$, then
    \begin{description}
    \item[composition($\circ$)] for each $c_1 \in V_1$ there
    exist $c_2 \in V_2$ such that both
      $(a_1,c_1,a_2,c_2)$ and $(c_1,b_1,c_2,b_2)$ are in $Z$;
    \item[identity($\mathrm{id}$)] if $a_1=b_1$ then $a_2=b_2$;
    \item[converse($\otimes$)] $(b_1,a_1,b_2,a_2) \in Z$;
    \end{description}
  \item[Back] if $(a_1,b_1,a_2,b_2)$ is in $Z$, then
    \begin{description}
    \item[composition($\circ$)] for each $c_2 \in V_2$ there
    exist $c_1 \in V_1$ such that both
      $(a_1,c_1,a_2,c_2)$ and $(c_1,b_1,c_2,b_2)$ are in $Z$;
    \item[identity($\mathrm{id}$)] if $a_2=b_2$ then $a_1=b_1$;
    \item[converse($\otimes$)] $(b_1,a_1,b_2,a_2) \in Z$.
    \end{description}
  \end{description}
  A marked structure $\mG_1 = (\G_1, a_1,b_1)$ is said to be
  \emph{arrow-logic bisimilar} to a
  marked structure $\mG_2 = (\G_2,a_2,b_2)$ if there is an
  arrow-logic bisimulation
  $Z$ between $\G_1$ and $\G_2$ containing $(a_1,b_1,a_2,b_2)$.
\end{definition}

The following characterization is now given by the
Hennessy-Milner theorem \cite[Theorem 2.24]{blackburn_modallogic}:
\begin{proposition}
  Let $\mG_1 = (\G_1,a_1,b_1)$ and $\mG_2=(\G_2,a_2,b_2)$ be
  \emph{finite} marked structures.
  Then $$ \mG_1 \indist \ffull \mG_2 \quad \Leftrightarrow \quad
\text{$\mG_1$ is arrow-logic bisimilar to $\mG_2$}. $$
\end{proposition}

\newcommand{\fsafe}{\F_{\rm safe}}
In the field of databases \cite[Chapter~5]{ahv_book},
it is good practice to employ
``safe'' query languages, meaning that 1 is not expressible,
i.e., queries are domain-independent.
This would mean replacing the complementation
operator by the difference operator, and removing the diversity
relation, leading to the fragment $\C(\fsafe)$, with $\fsafe$
consisting of difference and converse.
Furthermore, in database theory
much attention is being paid to ``positive'' query
languages, i.e., without the difference operator.  Since it is
still important to understand the distinction between safe and
unsafe query languages, one might add the diversity relation back
in, which would lead one to fragments such as
$\C({}^{-1},\di)$.  (Note that adding diversity to
$\fsafe$ would bring us back to the full calculus, since 1 is
expressible as $1' \cup \di$ and then complement $e^c$ as $1-e$.)
Also, one may be interested in understanding the power of
following relations backwards, and study fragments where converse
is removed, such as
$\C(-)$ or $\C(\di)$.  Then again one may add derived operations
(projection, coprojection, residuals) that become primitive in
specific fragments, leading to new fragments such
as $C(-,\lres,\rres)$ or $\C({}^{-1},\cpi)$.

For some fragments discussed above, the Hennessy-Milner theorem
adapts easily.  As a case in point, consider the positive fragment with
diversity, $\C({}^{-1},\di)$.  To account for the absence of
complementation, it suffices in the definition of bisimulation to
remove the Back condition, thus obtaining a kind of
\emph{simulation} rather than bisimulation.
To account for the diversity relation, it suffices to add it as a nullary
modality by adding the following
part to the forth-condition:
\begin{description}
\item[diversity($\mathrm{di}$)] if $a_1 \neq b_1$, then
$a_2 \neq b_2$.
\end{description}
We can then analogously show that $(\G_1,a_1,b_1)
\contnd {{}^{-1},\di} (\G_2,a_2,b_2)$ if and only if there exists
such a diversity-simulation
from $\G_1$ to $\G_2$ containing $(a_1,b_1,a_2,b_2)$.

Many other fragments, however, require much less obvious
adaptations to the notion of bisimulation for arrow logic.  For
instance, the coprojection and residual operations cannot simply
be considered to be extra modalities in arrow logic.  Another
difficulty arises when we remove the converse operator or the
diversity relation.  Expressions in such fragments always return
paths in the graph formed by the atomic steps
(Proposition~\ref{leminpaths}).  It does not suffice now to remove
the converse-forth or the diversity-forth parts in the definition
of bisimulation;  we also need to adapt the composition-forth
part.

In the remainder of this paper, we will show how appropriate
notions of simulation and bisimulation can be defined for all
fragments $\C(\F)$ of the calculus of relations considered in
this paper. 

\subsection{Some examples}
We conclude the present section with a few examples of
(in)distinguishability.  Thereto we introduce five example
structures in the vocabulary of a single relation name $R$,
shown in Figure~\ref{figraphs}.

\begin{figure}
\begin{center}
\begin{tabular}{c@{\hspace{5em}}c}
\begin{tabular}{c}
$\G_5$
\\[2ex]
\begin{tikzpicture}
\node (1) at (0,0) [scale=0.6,circle,fill,label=below:{$1$}] {};
\node (2) at (1,0) [scale=0.6,circle,fill,label=below:{$2$}] {};
\node (3) at (2,0) [scale=0.6,circle,fill,label=below:{$3$}] {};
\node (4) at (3,0) [scale=0.6,circle,fill,label=below:{$4$}] {};
\node (5) at (4,0) [scale=0.6,circle,fill,label=below:{$5$}] {};
\path (1) edge[->] (2) 
(2) edge[->] (3) 
(3) edge[->] (4) 
(4) edge[->] (5) ;
\end{tikzpicture}
\end{tabular}
&
\begin{tabular}{c}
$\G_4$
\\[2ex]
\begin{tikzpicture}
\node (2) at (1,0) [scale=0.6,circle,fill,label=below:{$2$}] {};
\node (3) at (2,0) [scale=0.6,circle,fill,label=below:{$3$}] {};
\node (4) at (3,0) [scale=0.6,circle,fill,label=below:{$4$}] {};
\node (5) at (4,0) [scale=0.6,circle,fill,label=below:{$5$}] {};
\path (2) edge[->] (3) 
(3) edge[->] (4) 
(4) edge[->] (5) ;
\end{tikzpicture}
\end{tabular}
\\[9ex]
\begin{tabular}{c}
$\G_9$
\\[2ex]
\begin{tikzpicture}
\node (1) at (0,0) [scale=0.6,circle,fill,label=below:{$1$}] {};
\node (2) at (1,0) [scale=0.6,circle,fill,label=below:{$2$}] {};
\node (3) at (2,0) [scale=0.6,circle,fill,label=below:{$3$}] {};
\node (4) at (3,0) [scale=0.6,circle,fill,label=below:{$4$}] {};
\node (5) at (4,0) [scale=0.6,circle,fill,label=below:{$5$}] {};
\node (6) at (0.8,1) [scale=0.6,circle,fill,label=below:{$6$}] {};
\node (7) at (1.6,1) [scale=0.6,circle,fill,label=below:{$7$}] {};
\node (8) at (2.4,1) [scale=0.6,circle,fill,label=below:{$8$}] {};
\node (9) at (3.2,1) [scale=0.6,circle,fill,label=below:{$9$}] {};
\path (1) edge[->] (2) 
(2) edge[->] (3) 
(3) edge[->] (4) 
(4) edge[->] (5)
(1) edge[->] (6)
(6) edge[->] (7)
(7) edge[->] (8)
(8) edge[->] (9)
(9) edge[->] (5) ;
\end{tikzpicture}
\end{tabular}
&
\begin{tabular}{c}
$\HG_4$
\\[2ex]
\begin{tikzpicture}
\node (1) at (0,0) [scale=0.6,circle,fill,label=below:{$1$}] {};
\node (2) at (1,-0.5) [scale=0.6,circle,fill,label=below:{$2$}] {};
\node (3) at (2,-0.5) [scale=0.6,circle,fill,label=below:{$3$}] {};
\node (4) at (1,0.5) [scale=0.6,circle,fill,label=below:{$4$}] {};
\path (1) edge[->] (2)
(2) edge[->] (3)
(1) edge[->] (4) ;
\end{tikzpicture}
\end{tabular}
\\[9ex]
&
\begin{tabular}{c}
$\HG_3$
\\[2ex]
\begin{tikzpicture}
\node (1) at (0,0) [scale=0.6,circle,fill,label=below:{$1$}] {};
\node (2) at (1,-0.5) [scale=0.6,circle,fill,label=below:{$2$}] {};
\node (3) at (2,-0.5) [scale=0.6,circle,fill,label=below:{$3$}] {};
\path (1) edge[->] (2)
(2) edge[->] (3)
(1) edge[->,out=30,in=135] (3) ;
\end{tikzpicture}
\end{tabular}
\end{tabular}
\end{center}
\caption{Five structures shown as directed graphs.}
\label{figraphs}
\end{figure}
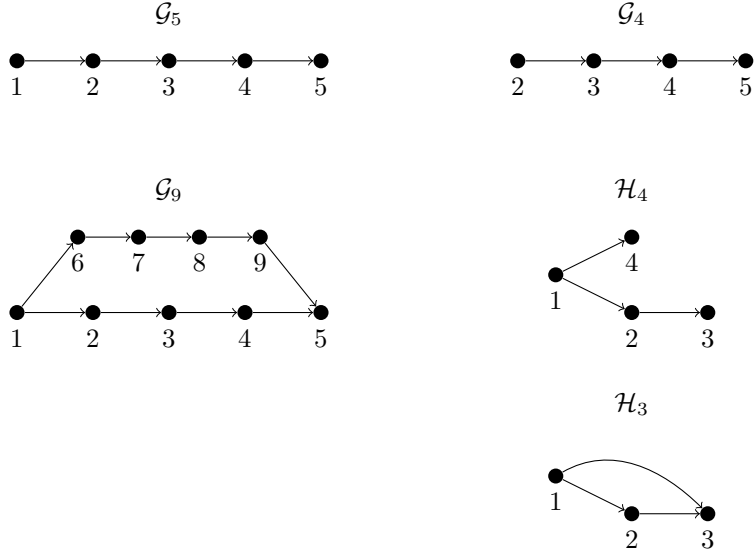

Let us begin by comparing the marked structures $(\G_5,2,5)$ and
$(\G_4,2,5)$.  They are distinguishable in the full calculus, for
example, by the expression\footnote{Recall that projection is
expressible in the full calculus.} $e_1 = \pi_2(R) \circ R^4$
which belongs to $\type \ffull{\G_5,2,5}$ but not to $\type
\ffull{\G_4,2,5}$.  In contrast, we have $(\G_5,2,5) \indist -
(\G_4,2,5)$, i.e., the two marked structures are
indistinguishable in $\C(-)$.  The intuition, following
Proposition~\ref{leminpaths}, is that expressions in this fragment
that return a pair $(a,b)$ on a structure are confined to the
part of the structure formed by all directed paths from $a$ to
$b$; on $\G_5$ and $\G_4$, the parts between $2$ and $5$ are
identical.  One way to make such an indistinguishability claim
formal will be our main result, which includes a notion of
bisimulation appropriate for the fragment $\C(-)$.

Consider now a positive fragment such as $\C(\pi)$. We obviously have
$(\G_5,2,5) \not \contnd \pi (\G_4,2,5)$ since the expression
$e_1$ above belongs to $\C(\pi)$.  Nevertheless, in the other
direction we do have $(\G_4,2,5) \contnd \pi (\G_5,2,5)$, as again
will follow from our main result.  On the other hand, in the
fragment $\C(\cpi)$ we have $(\G_4,2,5) \not \contnd \cpi
(\G_5,2,5)$ by the expression $\cpi_2(R) \circ R^4$.

We can illustrate degree-bounded indistinguishability by
$(\G_5,1,5) \indistk \fsafe 2 (\G_9,1,5)$.  Indeed, degree-two
expressions in $\C(\fsafe)$ are limited to paths of length at most
four, and both structures are identical inasfar as such paths between 1 and
5 are concerned.  Of course, the two marked structures are
  distinguishable in degree three; by the expression\footnote{$R^5$
  is expressible in degree three as $((R \circ R) \circ R) \circ (R
  \circ R)$.} $R^4 \cap R^5$ we already have distinguishability $
  (\G_9,1,5) \not \contndk{}3 (\G_5,1,5) $ in the
  most basic fragment $\C$.

Finally let us compare the marked structures $(\HG_4,1,1)$ and
$(\HG_3,1,1)$.  Clearly, in $\C(\pi)$, we have $(\HG_3,1,1) \not
\contndk \pi 1 (\HG_4,1,1)$ as witnessed by the expression
$\pi_1(R^2 \cap R)$.  In the other direction, however, we even
have $(\HG_4,1,1) \contnd {{}^{-1}} (\HG_3,1,1)$.  On the other hand,
adding diversity, we have $(\HG_4,1,1) \not \contnd {{}^{-1},\di}
(\HG_3,1,1)$ by the expression $R \circ ((R \circ \di) \cap \di)
\circ R^{-1}$.

\section{Bisimilarity and indistinguishability} \label{secbisim}

In this and the following section, we proceed as announced in the
preceding sections and define, for any fragment $\C(\F)$ and any
natural number $k$, an appropriate notion of bisimulation or
simulation between structures.  We will then show the adequacy of
the proposed notions in capturing indistinguishability in
$\C(\F)_k$.

In the present section, we deal with fragments containing
complement or difference.

\subsection{General definition of $(\F,k)$-bisimulation}
\label{subconditions}

\newcommand{\element}{(a_1,b_1,a_2,b_2)}

Let $\G_1$ and $\G_2$ be two structures with node sets $V_1$ and
$V_2$ respectively.  Let $\element$ be an arbitrary element of
$V_1^2 \times V_2^2$, and let $Z$ be an arbitrary subset of the
same set $V_1^2 \times V_2^2$.

For any given calculus fragment $\C(\F)$ that contains complement
or difference, we define the following suite of conditions.  The
conditions are relative to $\F$ in that they refer to $\atypeq
\F$ and $\paths^\F$.  All the conditions are also clearly
relative to $\G_1$ and $\G_2$.

\begin{description}

\item[Atoms Forth] We say that $\element$ has the \emph{Atoms
Forth property} if $$ \atype
\F{\G_1,a_1,b_1} \subseteq \atype \F{\G_2,a_2,b_2}. $$

\item[Atoms Back] We say that $\element$ has the \emph{Atoms
Back property} if $$ \atype
\F{\G_1,a_1,b_1} \supseteq \atype \F{\G_2,a_2,b_2}. $$

\end{description}

Furthermore, let $i>0$ be an arbitrary natural number.

\begin{description}

\item[Composition Forth] We say that $\element$
has the \emph{Composition Forth property at degree $i$ with
respect to $Z$} if for every $c_1$ in
$V_1$ with $(a_1,c_1)$ and $(c_1,b_1)$ in $\paths_{i-1}^\F(\G_1)$,
there exists $c_2$ in $V_2$ such that both $(a_1,c_1,a_2,c_2) \in
\Zet$ and $(c_1,b_1,c_2,b_2) \in \Zet$.

\item[Composition Back]
We say that $\element$ has the
\emph{Composition Back property
at degree $i$ with respect to $Z$} if for every
$c_2$ in $V_2$ with $(a_2,c_2)$ and $(c_2,b_2)$ in
$\paths_{i-1}^\F(\G_2)$, there exists $c_1$ in $V_1$ such that both
$(a_1,c_1,a_2,c_2) \in \Zet$ and $(c_1,b_1,c_2,b_2) \in
\Zet$.

\item[Projection Forth] We say that $\element$ has the
\emph{Projection Forth property at degree $i$ with respect to
$Z$} if either $a_1 \neq b_1$, or $a_1=b_1$ and $a_2=b_2$ and for
every $c_1$ in $V_1$ with $(a_1,c_1)$ in $\paths_{i-1}^\F(\G_1)$,
there exists $c_2$ in $\V_2$ such that $(a_1,c_1,a_2,c_2) \in Z$.
Moreover, if $a_1=b_1$ then also for every $c_1$ in $V_1$ with
$(c_1,a_1)$ in $\paths_{i-1}^\F(\G_1)$, there must exist $c_2$ in
$\V_2$ such that $(c_1,a_1,c_2,a_2) \in Z$.

\item[Projection Back] We say that $\element$ has the
\emph{Projection Back property at degree $i$ with respect to $Z$}
if either $a_2 \neq b_2$, or $a_2=b_2$ and $a_1=b_1$ and for
every $c_2$ in $V_2$ with $(a_2,c_2)$ in $\paths_{i-1}^\F(\G_2)$,
there exists $c_1$ in $\V_1$ such that $(a_1,c_1,a_2,c_2) \in Z$.
Moreover, if $a_2=b_2$, then also for every $c_2$ in $V_2$ with
$(c_2,a_2)$ in $\paths_{i-1}^\F(\G_2)$, there must exist $c_1$ in
$\V_1$ such that $(c_1,a_1,c_2,a_2) \in Z$.

\item[Left Residual Forth] We say that $\element$ has the
\emph{Left Residual Forth property at degree $i$ with respect to
$Z$} if for every $c_2$ in $V_2$ with $(b_2,c_2)$ in
$\paths_{i-1}^\F(\G_2)$, there exists $c_1$ in $V_1$ such that both
$(b_1,c_1,b_2,c_2) \in \Zet$ and either $(a_1,c_1) \notin
\paths^\F_{i-1}(\G_1)$ or $(a_1,c_1,a_2,c_2) \in \Zet$.

\item[Left Residual Back] We say that $\element$ has the
\emph{Left Residual Back property at degree $i$ with respect to
$Z$} if for every $c_1$ in $V_1$ with $(b_1,c_1)$ in
$\paths^\F_{i-1}(\G_1)$, there exists $c_2$ in $V_2$ such that both
$(b_1,c_1,b_2,c_2) \in \Zet$ and either $(a_2,c_2) \notin
\paths^\F_{i-1}(\G_2)$ or $(a_1,c_1,a_2,c_2) \in \Zet$.

\item[Right Residual Forth] We say that $\element$ has the
\emph{Right Residual Forth property at degree $i$ with respect to
$Z$} if for every $c_2$ in $V_2$ with $(c_2,a_2)$ in
$\paths^\F_{i-1}(\G_2)$, there exists $c_1$ in $V_1$ such that both
$(c_1,a_1,c_2,a_2) \in \Zet$ and either $(c_1,b_1) \notin
\paths^\F_{i-1}(\G_1)$ or $(c_1,b_1,c_2,b_2) \in \Zet$.

\item[Right Residual Back] We say that $\element$ has the
\emph{Right Residual Back property at degree $i$ with respect to
$Z$} if for every $c_1$ in $V_1$ with $(c_1,a_1)$ in
$\paths^\F_{i-1}(\G_1)$, there exists $c_2$ in $V_2$ such that both
$(c_1,a_1,c_2,a_2) \in Z$ and either $(c_2,b_2) \notin
\paths^\F_{i-1}(\G_2)$ or $(c_1,b_1,c_2,b_2) \in \Zet$.

\end{description}

Now let $k$ be a natural number and
let $\bar Z = (\Zet_0, \Zet_1, \dots,
\Zet_k)$ be a decreasing sequence of relations with
$\Zet_0 \subseteq V_1^2 \times V_2^2$, decreasing in the sense
that $Z_i \subseteq Z_{i-1}$ for each $i\in\{1,\dots,k\}$.

We lift the above conditions to apply to such sequences $\bar Z$
as follows:
\begin{itemize}
\item
We say that $\bar Z$ has the Atoms Forth property if every
element of $Z_0$ has this property, and similarly for the Atoms
Back property.
\item

For any of the other properties (from Composition Forth to Right
Residual Back), we say that $\bar Z$ has a certain property if
for every $i\in \{1,\dots,k\}$, every element of $Z_i$ has that
property at degree $i$ with respect to $Z_{i-1}$.

\end{itemize}

We are finally ready for our main definition.  The following
definition is the most natural and easy to state, but we will see later in
Corollary~\ref{unnecessary} that in some cases, some of the conditions
are actually redundant.

\begin{definition} \label{defbisim}
We call $\bar Z$ an \emph{$(\F,k)$-bisimulation from $\G_1$
to $\G_2$} if $\bar Z$ has
\begin{itemize}
\item
the Atoms Forth and Back properties;
\item
the Composition Forth and Back properties;
\item
the Projection Forth and Back properties when $\F$ contains
projection;
\item
the Left (Right) Residual Forth and Back properties when $\F$
contains left (right) residual.
\end{itemize}

Given two marked structures $\mG_1=(\G_1,a_1,b_1)$ and
$\mG_2=(\G_2,a_2,b_2)$, when
there exists an $(\F,k)$-bisimulation $\bar Z$ from $\G_1$ to
$\G_2$ such that $(a_1,b_1,a_2,b_2) \in Z_k$,
we say that $\mG_1$ and $\mG_2$ are
\emph{$(\F,k)$-bisimilar}, and denote this by $\mG_1 \bisimeq \F k
\mG_2$.
\end{definition}

Note that in any bisimulation $\bar Z$, since both the Forth
and Back versions of the Atoms property must be satisfied, and
since identity is always present in our fragments, for each
$(a_1,b_1,a_2,b_2) \in Z_0$ we have $a_1=b_1$ iff $a_2=b_2$.
From this it follows that $a_1 \neq b_1$ iff $a_2 \neq b_2$, so
it does not matter whether or not $\di$ belongs to $\F$.  Hence, 
there really are
only two variants of the Atoms properties, depending on whether or
not $\F$ contains converse.  Both variants stipulate
$a_1=b_1$ iff $a_2=b_2$ as just seen.  When $\F$ does not contain
converse, the Atoms conditions stipulate furthermore that exactly
the same relations $R$ from $\Lambda$ must hold for $(a_1,b_1)$
and $(a_2,b_2)$; moreover, when $\F$ does contain converse, also
the converse relations that hold must be the same.

One may wonder why there are no conditions corresponding to the
coprojection operation.  The reason is that in fragments
containing difference, coprojection plays no additional role
beyond that of projection, because $\cpi(e) \equiv \id - \pi(e)$.
Formally, we will see in the Invariance Lemma that coprojections
are also preserved by bisimulations.

\subsection{Adequacy theorem}

We establish:

\begin{theorem}[Adequacy Theorem] \label{water}
For any fragment $\C(\F)$ where $\F$ contains complement or
difference, we have
$\mG_1 \bisimeq \F k \mG_2$ if and only if $\mG_1 \indistk \F k \mG_2$.
\end{theorem}

We will prove the only-if direction
in the Invariance Lemma, where we will need the following
immediate property:

\begin{lemma} \label{leminduction}
Let $k>0$ and let $(Z_0,Z_1,\dots,Z_k)$ be an
$(\F,k)$-bisimulation from $\G_1$ to $\G_2$. Then
$(Z_0,\dots,Z_{k-1})$ is an $(\F,k-1)$-bisimulation from $\G_1$ to $\G_2$.
\end{lemma}

We now give:

\begin{lemma}[Invariance] \label{invariant}
If $\mG_1 \bisimeq \F k \mG_2$ then $\mG_1 \indistk \F k \mG_2$.
\end{lemma}
\begin{proof}
Let $e$ be an expression in $\C(\F)_k$; we may assume by
Proposition~\ref{convprop} that converse is only applied to
relation names.
We prove by induction on the structure
of $e$ that for any marked structures
$\mG_1 = (\G_1, a_1,b_1) \bisimeq \F k \mG_2 = (\G_2, a_2, b_2)$, we have
$(a_1,b_1) \in e(\G_1)$ if and only if $(a_2,b_2) \in e(\G_2)$.

Let $\V_1$ and $\V_2$ be the node sets of the structures $\G_1$
and $\G_2$, respectively.  Let $\bar Z$ be an
$(\F,k)$-bisimulation from $\G_1$ to $\G_2$ such that
$(a_1,b_1,a_2,b_2) \in Z_k$.

For the base case, where $e$ is an atomic expression,
the result follows immediately from the Atoms Forth and Back
properties applied to $(a_1,b_1,a_2,b_2) \in Z_k$.

If $e$ is $e_1 \cup e_2$, $e_1 \cap e_2$, $e_1 - e_2$, or
$\compl{e_1}$, the result follows immediately from the induction
hypothesis.

For the case where $e$ is $e_1 \compos e_2$, consider the
only-if, i.e., assume that $(a_1,b_1) \in e(\G_1)$. By definition
of composition, there exists $c_1$ in $\V_1$ with $(a_1,c_1) \in
e_1(\G_1)$ and $(c_1,b_1) \in e_2(\G_1)$. Since $e_1$ and $e_2$
have depth at most $k-1$, by Proposition~\ref{leminpaths}, both
$(a_1,c_1)$ and $(c_1,b_1)$ are in $\paths_{k-1}^\F(\G_1)$. By
the Composition Forth condition, there exists $c_2$ in $\V_2$
such that both $(a_1,c_1,a_2,c_2)$ and $(c_1,b_1,c_2,b_2)$ belong
to $Z_{k-1}$.  Hence, by induction (and using
Lemma~\ref{leminduction}), we have $(a_2,c_2) \in e_1(\G_2)$ and
$(c_2,b_2) \in e_2(\G_2)$, whence $(a_2,b_2) \in e_1 \compos
e_2(\G_2)$. The argument for the if-direction is similar, using
Composition Back instead of Composition Forth.

For the case where $e$ is $\pi_1(e_1)$, consider the only-if,
i.e., assume $(a_1,b_1) \in \pi_1(e_1)(\G_1)$. By definition of
projection, we have $a_1 = b_1$, and there exists $c_1$ in $\V_1$
with $(a_1,c_1) \in e_1(\G_1)$. By Proposition~\ref{leminpaths}, we
have $(a_1,c_1)$ in $\paths_{k-1}^\F(\G_1)$. By the Projection
Forth condition, there exists $c_2$ in $\V_2$ such that
$(a_1,c_1,a_2,c_2) \in Z_{k-1}$.  Hence, by induction, we have
$(a_2,c_2) \in e_1(\G_2)$, whence $(a_2,b_2) \in
\pi_1(e_1)(\G_2)$. The argument for the if-direction is similar,
using Projection Back instead of Projection Forth. The argument
for the case where $e$ is $\pi_2(e_1)$ is analogous.

The case where $e$ is $\cpi_i(e_1)$ now follows readily from the
preceding as $\cpi_i(e_1) \equiv \id - \pi_i(e_1)$.

Finally, for the case where $e$ is $e_1 \lres e_2$, consider the only-if,
i.e., assume $(a_1,b_1) \in e(\G_1)$. Suppose now that $(a_2,b_2)
\notin e_1 \lres e_2(\G_2)$. Then, by definition of the left
residual, there exists $c_2$ in $\V_2$ such that $(b_2,c_2)
\in e_2(\G_2)$ and $(a_2,c_2) \notin e_1(\G_2)$.
By Proposition~\ref{leminpaths}, we
have that $(b_2,c_2)$ in $\paths_{k-1}^\F(\G_2)$. By the Left
Residual Forth condition,
there exists $c_1$ in $\V_1$ such that both $(b_1,c_1,b_2,c_2)$
and $(a_1,c_1,a_2,c_2)$ belong to $Z_{k-1}$.
Hence, by induction, we
obtain $(b_1,c_1) \in e_2(\G_1)$ and $(a_1,c_1) \notin
e_1(\G_1)$. Now, this $c_1$ contradicts that $(a_1,b_1) \in e_1
\lres e_2 (\G_1)$. The
argument for the if-direction is similar, using
Left Residual Back.

The case of a right residual is completely analogous to that of a
left residual.
\end{proof}

The other direction of the adequacy theorem will be established
by the Characteristic Expression Lemma.  For the proof of that Lemma we need
to introduce the construction of the maximal bisimulation.  This
construction by successive refinement is a classical technique
\cite[Section~3.5]{gorankotto} and is used as well in algorithms
for classical bisimilarity \cite{bisimalgo} and for color
refinement \cite{immermanlander}.

\begin{definition} \label{defmax}
Given a fragment $\F$ as above and structures
$\G_1$ and $\G_2$ with node sets $V_1$ and $V_2$ respectively,
we construct an infinite decreasing sequence
$Z_0, Z_1, Z_2, \cdots$ by induction on $k$ as follows.
\begin{enumerate}
\item
$Z_0$ is the set of all elements of $V_1^2 \times V_2^2$ that have the
Atoms Forth and Back properties
relative to $\F$, $\G_1$ and $\G_2$.
\item
$Z_i$, for $i>0$, is the set of all elements in $Z_{i-1}$ that
have
\begin{itemize}
\item
the Composition Forth and Back properties
at degree $i$ with respect to $Z_{i-1}$
(still relative to $\F$, $\G_1$ and $\G_2$);
\item
the Projection Forth and Back properties
at degree $i$ with respect to $Z_{i-1}$,
if $\F$ contains projection; and
\item
the Left (Right) Residual Forth and Back properties
at degree $i$ with respect to $Z_{i-1}$, if $\F$ contains left
(right) residual.
\end{itemize}
\end{enumerate}
We denote this constructed sequence by
$\BiSim \F {\G_1} {\G_2}$.
\end{definition}

The relevant property about 
$\BiSim \F {\G_1} {\G_2}$ is the following. It follows
immediately from the definitions.

\begin{proposition} \label{propmax}

Let $\BiSim \F {\G_1} {\G_2} = Z_0, Z_1, \dots$.
Then for each natural number $k$, the sequence $Z_0,
Z_1, \dots, Z_k$ is an
$(\F,k)$-bisimulation from $\G_1$ to $\G_2$.  Furthermore, it is
the maximal bisimulation in the sense that, for any other such
$(\F,k)$-bisimulation $Z'_0, \dots, Z'_k$, we
have $Z'_i \subseteq Z_i$ for each $i=0,\dots,k$.

\end{proposition}

As an immediate corollary, we have:

\begin{corol} \label{cormax}
$(\G_1,a_1,b_1) \bisimeq \F k
(\G_2,a_2,b_2) \; \Leftrightarrow \; (a_1,b_1,a_2,b_2) \in \BiSim \F
{\G_1} {\G_2}_k$.
\end{corol}

Another relevant property is the following.

\begin{proposition}
\label{pathpreserving}
Let $\BiSim \F {\G_1} {\G_2} = Z_0, Z_1, \dots$.
Then $Z_i$, for every $i$, is \emph{path-preserving at degree
$i$}, in the
sense that if $(a_1,b_1,a_2,b_2) \in Z_i$
and $(a_1,b_1) \in \paths^\F_i(\G_1)$, then
also $(a_2,b_2) \in \paths^\F_i(\G_2)$.
\end{proposition}
\begin{proof}  By induction on $i$.  The base case $i=0$ is
clear from the definition of $Z_0$.
For the case $i>0$ we may assume
that 1 is absent from $\F$, since otherwise the claim is trivial.
We can thus rely on the equivalence
$\paths^\F_i \equiv \paths^\F_{i-1} \cup
(\paths^\F_{i-1} \circ \paths^\F_{i-1})$.  If 
$(a_1,b_1) \in \paths^\F_{i-1}(\G_1)$, then, since
$Z_i \subseteq Z_{i-1}$, the claim follows
directly by induction.  Otherwise, there exists $c_1 \in V_1$ such
that $(a_1,c_1)$ and $(c_1,b_1)$ are in $\paths^\F_{i-1}(\G_1)$.
By the Composition Forth Property, there exists $c_2 \in V_2$
such that $(a_1,c_1,a_2,c_2)$ and $(c_1,b_1,c_2,b_2)$ belong to
$Z$.  By induction, we have $(a_2,c_2)$ and $(c_2,b_2)$ in
$\paths^\F_{i-1}(\G_2)$, whence $(a_2,b_2) \in \paths^\F_i(\G_2)$
as desired.
\end{proof}

Note that, since the above proposition shows path-preservation
for the maximal bisimulation, path-preservation also holds for
any arbitrary bisimulation.

We are now ready for:

\begin{lemma}[Characteristic Expression] \label{karakter}
Let $k$ be a natural number and let $\mG_1 = (\G_1, a_1, b_1)$ be
a marked structure.  Then there
exists an expression $e^{\F,k}_{\mG_1}$ in $\C(\F)_k$ such that
for every structure $\G_2$ we have
\[ e^{\F,k}_{\mG_1}(\G_2) = \{
(a_2,b_2) \in \paths^\F_k(\G_2) \mid \mG_1 \bisimeq \F k
(\G_2,a_2,b_2) \}.  \]
\end{lemma}
\begin{proof}
The construction of the required expression is by induction on $k$.
For the base of the construction we put
$ e^{\F,0}_{\mG_1} := \phiposatoms - \phinegatoms $,
where
  $$
\phiposatoms
:=
\bigcap_{e \in \atype \F{\mG_1}} e
\qquad \text{and} \qquad
\phinegatoms
:=
\bigcup_{e \in \aexp(\F) - \atype \F{\mG_1}} e.
$$
It is clear that $(a_2,b_2) \in e^{\F,0}_{\mG_1}(\G_2)$ iff 
$\atype \F{\mG_1} = \atype \F{\mG_2}$, which is necessary and
sufficient for $\mG_1 \bisimeq \F 0 (\G_2,a_2,b_2)$ to hold.

For the inductive step of the construction, let $k>0$. Our
approach, based on Corollary~\ref{cormax}, is to
show that each of the properties $P$ involved in
Definition~\ref{defmax},
$P$ ranging from Composition Forth
until Right Residual Back, is expressible in $\C(\F)_k$, in the
following way.  Let $V_1$ be the node set of the fixed
structure $\G_1$, and let $V_2$ be the node set of any structure
$\G_2$ to which our expressions will be applied.  We may
represent a set $Z \subseteq V_1^2 \times V_2^2$ by the family of
binary relations consisting of, for each $(a,b) \in V_1^2$, the
binary relation  $$ Z_{a,b} =
\{(a',b') \in \faths {k-1}(\G_2) \mid (a,b,a',b') \in Z\}. $$
Note that these relations make only the ``slice'' of $Z$ visible
governed by the constraint $(a',b') \in \faths {k-1}(\G_2)$.

We can add all these relations to $\G_2$, yielding an expanded
structure, denoted by $(\G_2,Z)$, over the expansion of the given
vocabulary $\Lambda$ with relation names $Z_{a,b}$ for each
$(a,b) \in V_1^2$.  (Here, we are abusing notation a bit by
making no formal distinction between the relation name and its
contents.) Furthermore, for our purpose, it will be sufficient to
assume that $Z$ is path-preserving at degree $k-1$ in the sense
of Proposition~\ref{pathpreserving}.  Now we are going to express
each property $P$ by an expression $\phigenp$ over the expanded
vocabulary, in the sense for any $\G_2$ and any $Z$ as above,
$\phigenp$ applied to $(\G_2,Z)$ returns the set of pairs
$(a_2,b_2) \in \paths^\F_k(\G_2)$ for which $(a_1,b_1,a_2,b_2)$
satisfies the $P$ property at degree $k$ with respect to $Z$,
relative to $\G_1$, $\G_2$, and $\F$.

There is a caveat:  if $\G_1$ is infinite, there are infinitely
many pairs $(a,b) \in V_1^2$, so also infinitely many relations
$Z_{a,b}$.  Accordingly, we will allow the expression
$\phigenp$ to be infinitary, in that it can use infinite
unions and intersections.   We will see later in the inductive
argument that normal, finitary expressions are still obtained in
the end.

For example, for the Composition Forth property, we have the
following expression, whose correctness is evident:
\begin{flushleft}
\quad
$\displaystyle
\phicompforth := \faths k \cap \bigcap_{\substack{c_1 \in \V_1\\(a_1,c_1) \in
            \paths_{k-1}^\F(\G_1)\\(c_1,b_1) \in \paths_{k-1}^\F(\G_1)}}
Z_{a_1,c_1} \circ Z_{c_1,b_1}.
$
\end{flushleft}
Here and below, it is understood that an empty intersection
vanishes from the expression.  Empty unions, as usual, are
replaced by the expression 0.

For the Projection Forth property, if $a_1 \neq b_1$, we can simply define
$$\phiprojforth := \faths k.$$  Otherwise, if $a_1 = b_1$,
we use
\begin{multline*}
\phiprojforth := \id \\
{} \cap \bigcap_{\substack{c_1 \in \V_1\\(a_1,c_1) \in
            \paths_{k-1}^{\F}(\G_1)}} \pi_1(Z_{a_1,c_1})
	    \; \cap \bigcap_{\substack{c_1 \in \V_1\\(c_1,a_1) \in
            \paths_{k-1}^{\F}(\G_1)}}
	    \pi_2(Z_{c_1,a_1}).
\end{multline*}
Again the correctness is evident.

For the Projection Back property, if $a_1 \neq b_1$, we can
simply define
$$\phiprojback := \faths k - \id .$$
Indeed, if $(a_2,b_2) \in (\faths k - \id)(\G_2)$ then $a_2 \neq
b_2$, in which case $(a_1,b_1,a_2,b_2)$ voidly satisfies the Projection Back
property.  Conversely, if $(a_2,b_2) \in \faths k(\G_2)$ and
$(a_1,b_1,a_2,b_2)$ satisfies the
Projection Back property, but $a_1\neq b_1$, then $a_2\neq b_2$
must hold as well so $(a_2,b_2) \in (\faths k - \id)(\G_2)$ as
desired.

If $a_1 = b_1$, we put
\begin{multline*}
\phiprojback := \faths k \\
{} - \bigl ( \id \cap
(
\pi_1 ( \paths^{\F}_{k-1} - \bigcup_{c_1 \in \V_1} Z_{a_1,c_1})
\cup
\pi_2 ( \paths^{\F}_{k-1} - \bigcup_{c_1 \in \V_1} Z_{c_1,a_1})
)
\bigr ).
\end{multline*}
Let us verify the correctness in this case.  Let $(a_2,b_2) \in
\phiprojback(\G_2,Z)$.
Then clearly $(a_2,b_2) \in \faths k(\G_2)$.
We must show that
$(a_1,b_1,a_2,b_2)$ satisfies the
Projection Back property at degree $k$ with respect to $Z$.
If $a_2 \neq b_2$ this is trivial so assume $a_2 = b_2$.  Then
$(a_2,b_2) \in \id(\G_2)$, so, considering the expression
$\phiprojback$, this means that (i) $(a_2,a_2) \notin 
\pi_1 ( \paths^{\F}_{k-1} - \bigcup_{c_1 \in \V_1}
Z_{a_1,c_1})(\G_2,Z)$ and (ii)
$(a_2,a_2) \notin 
\pi_2 ( \paths^{\F}_{k-1} - \bigcup_{c_1 \in \V_1}
Z_{c_1,a_1})(\G_2,Z)$.  Now let $c_2 \in V_2$ such that $(a_2,c_2)
\in \faths {k-1}(\G_2)$.  We must show there exists $c_1 \in V_1$ such that
$(a_1,c_1,a_2,c_2) \in Z$.  For the sake of contradiction, assume
the contrary; then $(a_2,c_2) \in (\faths {k-1} - \bigcup_{c_1 \in V_1}
Z_{a_1,c_1})(\G_2,Z)$.  By (i), this is impossible.  Similarly,
using (ii), we obtain that for any $c_2 \in V_2$ such that
$(c_2,a_2) \in \faths {k-1}(\G_2)$ there exists $c_1 \in V_1$
such that $(c_1,a_1,c_2,a_2) \in Z$.  Hence the Projection Back
property holds as desired.

The converse direction, that any $(a_2,b_2) \in \faths k(\G_2)$
must belong to
$\phiprojback(\G_2,Z)$ if $(a_1,b_1,a_2,b_2)$
satisfies the Projection Back property at degree $k$ with respect
to $Z$, is argued similarly.

For the Composition Back property, the expression is a little bit
less evident:
\begin{multline*}
\phicompback
    :=
    \paths^\F_{k} \\
    {} - \bigcup_{V \subseteq \V_1} \bigl (
  (\paths^\F_{k-1} - \bigcup_{c_1 \in \V} Z_{a_1,c_1})
  \; \compos \;
  (\paths^\F_{k-1} - \bigcup_{c_1 \in \V_1 - V} Z_{c_1,b_1}) \bigr )
\end{multline*}

To see the correctness of the above expression, we must show for
any $(a_2,b_2) \in \paths^\F_k(\G_2)$ that $(a_2,b_2) \in
\phicompback(\G_2,Z)$ if and only if $(a_1,b_1,a_2,b_2)$
has the Composition Back property at degree $k$ with respect to $Z$.  For
the only-if direction, assume $(a_2,b_2) \in
\phicompback(\G_2,Z)$.  Let $c_2 \in V_2$ such that both
$(a_2,c_2)$ and $(c_2,b_2)$ are in $\paths^\F_{k-1}(\G_2)$.  We must
show that there exists $c_1 \in V_1$ such that both
$(a_1,c_1,a_2,c_2)$ and $(c_1,b_1,c_2,b_2)$ belong to $Z$.  For
the sake of contradiction, suppose such $c_1$ does not exist,
i.e., for each $c_1 \in V_1$ either $(a_2,c_2) \notin
Z_{a_1,c_1}$ or $(c_2,b_2) \notin Z_{c_1,b_1}$ (or both).  Then,
letting $V := \{c_1 \in V_1 \mid (a_2,c_2) \notin Z_{a_1,c_1}\}$,
we have for any $c_1 \in V_1-V$ that $(c_2,b_2) \notin
Z_{c_1,b_1}$.  Hence, through $c_2$, we see that $$ (a_2,b_2) \in
( \paths^\F_{k-1}(\G_2) - \bigcup_{c_1\in V} Z_{a_1,c_1}) \; \circ
( \paths^\F_{k-1}(\G_2) - \bigcup_{c_1 \in V_1-V} Z_{c_1,b_1}), $$
which contradicts $(a_2,b_2) \in \phicompback(\G_2,Z)$.

For the if-direction, assume that $(a_1,b_1,a_2,b_2)$ has
the Composition Back property at degree $k$ with respect to $Z$, and we must
show that $(a_2,b_2)$ belongs to $\phicompback(\G_2,Z)$.  For the sake of
contradiction, assume there exists $V \subseteq V_1$ and $c_2 \in
V_2$ such that $(a_2,c_2) \in \paths^\F_{k-1}(\G_2) - Z_{a_1,c_1}$
for all $c_1\in V$,
and $(c_2,b_2) \in \paths^\F_{k-1}(\G_2) - Z_{b_1,c_1}$ for all $c_1 \in V_1-V$.
This simply means that there exists $c_2 \in V_2$ such that
$(a_2,c_2)$ and $(c_2,b_2)$ both belong to $\paths^\F_{k-1}(\G_2)$,
but for which there exists no
$c_1 \in V_1$ such that both $(a_1,c_1,a_2,c_2)$ and $(c_1,b_1,c_2,b_2)$
belong to $Z$.  Thus we obtain a direct contradiction with the Composition Back
property.

The expressions for the Left and Right Residual Back properties
are more straightforward again:

\begin{flushleft}
\quad
$\displaystyle
\philresback
    := \paths^\F_k \; -
        \bigcup_{\substack{c_1\in \V_1 \\ (b_1,c_1) \in
	\paths^\F_{k-1}(\G_1)}} (\paths^\F_{k-1} - Z_{a_1,c_1})
        \lres Z_{b_1,c_1} 
$
\end{flushleft}
\begin{flushleft}
\quad
$\displaystyle
\phirresback
    := \paths^\F_k \; -
        \bigcup_{\substack{c_1\in \V_1 \\
	(c_1,a_1) \in \paths^\F_{k-1}(\G_1)}}
	Z_{c_1,a_1} \rres (\paths^\F_{k-1} - Z_{c_1,b_1})
$
\end{flushleft}

Let us show the correctness of the Left Residual Back expression;
the argument for the Right Residual is completely analogous.  Let
$(a_2,b_2) \in \faths k(\G_2)$.  We see that $(a_2,b_2) \in
\philresback$ if and only if there does \emph{not} exist $c_1 \in
V_1$ such that $(b_1,c_1) \in \faths {k-1}(\G_1)$, and such that
for each $c_2 \in V_2$ with $(b_2,c_2) \in Z_{b_1,c_1}$ we have
$(a_2,c_2) \in \faths {k-1}(\G_2) - Z_{a_1,c_1}$.  Equivalently,
$(a_2,b_2) \in \philresback$ iff for all $c_1 \in V_1$ with
$(b_1,c_1) \in \faths {k-1}(\G_1)$, there exists $c_2 \in V_2$
such that $(b_2,c_2) \in \faths {k-1}(\G_2)$ and
$(b_1,c_1,b_2,c_2) \in Z$ and $(a_2,c_2) \notin \faths
{k-1}(\G_2) - Z_{a_1,c_1}$.  By the path-preserving property of
$Z$, the qualification $(b_2,c_2) \in \faths {k-1}(\G_2)$ is
redundant.  Moreover, $(a_2,c_2) \notin \faths {k-1}(\G_2) -
Z_{a_1,c_1}$ means that either $(a_2,c_2) \notin \faths
{k-1}(\G_2)$ or $(a_1,c_1,b_1,c_2) \in Z$.  Thus we get exactly
the formulation of the Left Residual Back property.

Finally, the Left and Right Residual Forth properties are
expressed using a similar approach as for the Composition Back
property:

\begin{multline*}
\philresforth
    := \faths k \\
    {} \cap
    \bigcap_{\V \subseteq \V_1} \Bigl [
      ( \bigcup_{c_1 \in \V_1-\V}  Z_{a_1,c_1} )
\; \lres {} \\
( \paths_{k-1}^\F
 - \bigcup_{\substack{c_1 \in V_1 \\ (a_1,c_1) \notin
\paths^\F_{k-1}(\G_1)}} Z_{b_1,c_1}
	- \bigcup_{c_1 \in V} Z_{b_1,c_1} ) \Bigr ]
\end{multline*}
\begin{multline*}
\phirresforth
    := \faths k \\
    {} \cap
    \bigcap_{\V \subseteq \V_1} \Bigl [
      ( \paths_{k-1}^\F
      - \bigcup_{\substack{c_1 \in V_1 \\ (c_1,b_1) \notin \faths
	{k-1}(\G_1)}} Z_{c_1,a_1}
      - \bigcup_{c_1 \in \V} Z_{c_1,a_1} ) \\
        {} \rres \; \bigcup_{c_1 \in \V_1-\V} Z_{c_1,b_1} \Bigr ]
\end{multline*}

Let us show the correctness of the Left Residual Forth expression;
again the argument for the Right Residual is completely analogous.
We must show for any $(a_2,b_2) \in \paths^\F_k(\G_2)$ that
$(a_2,b_2) \in \philresforth(\G_2,Z)$ if and only if
$(a_1,b_1,a_2,b_2)$ has the Left Residual Forth property at
degree $k$ with respect to $Z$.  For the only-if direction,
assume $(a_2,b_2) \in \philresforth(\G_2,Z)$.  Let $c_2 \in V_2$
such that $(b_2,c_2) \in \paths^\F_{k-1}(\G_2)$.  We must show
that there exists $c_1 \in V_1$ such that $(b_1,c_1,b_2,c_2) \in
Z$ and either $(a_1,c_1) \notin \paths^\F_{k-1}(\G_1)$ or
$(a_1,c_1,a_2,c_2) \in Z$.  If there exists $c_1 \in V_1$ with
$(a_1,c_1) \notin \paths^\F_{k-1}$ such that $(b_1,c_1,b_2,c_2)
\in Z$, there is nothing to prove.  So, suppose no such $c_1$
exists, and consider $V := \{c_1 \in V_1 \mid (b_2,c_2) \notin
Z_{b_1,c_1}\}$.  Then $$ (b_2,c_2) \in \paths^\F_{k-1}(\G_2) -
\bigcup_{\substack{c_1 \in V_1 \\ (a_1,c_1) \notin
\paths^\F_{k-1}(\G_1)}} Z_{b_1,c_1} - \bigcup_{c_1 \in V}
Z_{b_1,c_1}.  $$ Hence, since $(a_2,b_2) \in 
\philresforth(\G_2,Z)$, we have $(a_2,c_2) \in \bigcup_{c_1 \in
\V_1-\V}  Z_{a_1,c_1}$, i.e., there exists $c_1 \in V$ such that
$(a_1,c_1,a_2,c_2) \in Z$ and $(b_1,c_1,b_2,c_2) \in Z$, as
desired.

For the if-direction, assume that $(a_1,b_1,a_2,b_2)$ has the
Left Residual Forth property at degree $k$ with respect to $Z$,
and let $V \subseteq V_1$ be arbitrary.  Let $c_2 \in V_2$ such that
$$ (b_2,c_2) \in \paths^\F_{k-1}(\G_2) -
\bigcup_{\substack{c_1 \in V_1 \\ (a_1,c_1) \notin
\paths^\F_{k-1}(\G_1)}} Z_{b_1,c_1} - \bigcup_{c_1 \in V}
Z_{b_1,c_1}. $$ By the Residual Forth property, there exists $c_1
\in V_1$ such that $(b_1,c_1,b_2,c_2) \in
Z$ (in particular, $c_1 \in V_1-V$)
and either $(a_1,c_1) \notin \paths^\F_{k-1}(\G_1)$ or
$(a_1,c_1,a_2,c_2) \in Z$.  The possibility $(a_1,c_1) \notin
\paths^\F_{k-1}(\G_1)$ cannot occur, however, because
$$ (b_2,c_2) \in \paths^\F_{k-1}(\G_2) -
\bigcup_{\substack{c_1 \in V_1 \\ (a_1,c_1) \notin
\paths^\F_{k-1}(\G_1)}} Z_{b_1,c_1}. $$  Hence, $(a_1,c_1) \in
\paths^\F_{k-1}(\G_1)$ and $(a_1,c_1,a_2,c_2) \in Z$, whence
$(a_2,c_2) \in Z_{a_1,c_1}$, since $Z$ was assumed to be path-preserving.
We conclude that $$(a_2,c_2) \in \bigcup_{c_1 \in V_1-V}
Z_{a_1,c_1}$$ as desired.

We are now ready to conclude the construction of 
the required expression $e^{\F,k}_{\mG_1}$.  This expression, applied to any
$\G_2$, should return the set of all pairs
$(a_2,b_2)$ such that $(a_1,b_1,a_2,b_2) \in \BiSim \F {\G_1} {\G_2}_k$.
By definition, these are the pairs $(a_2,b_2)$ for which
$(a_1,b_1,a_2,b_2) \in \BiSim \F {\G_1} {\G_2}_{k-1}$, and such
that $(a_1,b_1,a_2,b_2)$ has all the properties $P$, required by
Definition~\ref{defmax}, at degree $k$ with
respect to $\BiSim \F {\G_1} {\G_2}_{k-1}$.  We have just seen
that each such property $P$ is expressible by the infinitary
expression $\phigenp$.  Hence, we can obtain
$e^{\F,k}_{\mG_1}$ simply as the intersection of
$e^{\F,k-1}_{\mG_1}$ (obtained by induction)
and the expressions $\phigenp$ for the
different properties $P$ required by $(\F,k)$-bisimulation.

The only problem remaining is that each expression $\phigenp$
is still infinitary, and referring to extra
relation names of the form $Z_{a,b}$.
For our purpose, such a relation name should hold
the relation $\{(a',b') \in \faths {k-1}(\G_2) \mid (a,b,a',b') \in 
\BiSim \F {\G_1} {\G_2}_{k-1}\}$.  By the induction hypothesis
however, we can express this relation by the expression
$e^{\F,k-1}_{\G_1,a,b}$.  So, in $\phigenp$, we can replace each
occurrence of $Z_{a,b}$ by $e^{\F,k-1}_{\G_1,a,b}$ and obtain an
expression of degree $k$ over the original given vocabulary
$\Lambda$.

The resulting expression still has infinite
unions and intersections.  These unions and
intersections are over sets of expressions of degree $k$,
however.  Hence, since, up to equivalence, there are only finitely
many expressions of degree $k$ over the fixed finite vocabulary
$\Lambda$, we can equivalently replace the infinite unions and
intersections by finite ones.  The reason why there are only a
finite number of inequivalent expression of degree $k$ is the
same as why there are only a finite number of inequivalent
first-order logic formulas of quantifier rank $k$
\cite{eft_mathlogic}.

\end{proof}

We can conclude the proof of the Adequacy Theorem as follows:  

\begin{proof}[Proof of Theorem~\ref{water}]  The only-if
direction is proven by Lemma~\ref{invariant}.  The if-direction
for $k=0$ is clear.   So now assume
$(\G_1,a_1,b_1) \equiv^\F_k (\G_2,a_2,b_2)$ with $k>0$.
We distinguish two possibilities.  If
$(a_2,b_2) \notin \paths^\F_k(\G_2)$, then $\F$ cannot contain
the residuals, for otherwise $\paths^\F_k \equiv 1$.  (The only
exception is when $V_2$ is empty, but then $V_1$ must be empty as
well and the theorem becomes trivial.)  Moreover,
since $\paths^\F_k$ is expressible by
an expression of degree $k$ and
$(\G_1,a_1,b_1) \equiv^\F_k (\G_2,a_2,b_2)$, also 
$(a_1,b_1) \notin \paths^\F_k(\G_1)$.  But in that case we can see
that the Atoms Forth and Back, the Composition Forth and Back, as
well as the Projection Forth and Back properties are void, so
$(\G_1,a_1,b_1) \bisimeq \F k (\G_2,a_2,b_2)$ holds trivially.

Hence, the nondegenerate case is where
$(a_2,b_2) \in \paths^\F_k(\G_2)$ and
$(a_1,b_1) \in \paths^\F_k(\G_1)$, which allows us to invoke the
Charasteristic Expression Lemma.  We argue as follows.
First, we note that $(\G_1,a_1,b_1) \bisimeq \F k (\G_1,a_1,b_1)$
trivially holds; indeed, we can take the bisimulation $\bar Z$
where $Z_i = \{(a,b,a,b) \mid (a,b) \in V_1^2\}$ for each $i$.
Hence, by the Characteristic Expression Lemma, we have
$(a_1,b_1) \in e^{\F,k}_{\G_1,a_1,b_1}(\G_1)$.  Since
$(\G_1,a_1,b_1) \equiv^\F_k (\G_2,a_2,b_2)$, this implies
$(a_2,b_2) \in e^{\F,k}_{\G_1,a_1,b_1}(\G_2)$.  Again by the
Characteristic Expression Lemma this implies
$(\G_1,a_1,\allowbreak b_1) \bisimeq \F k (\G_1,a_1,b_1)$ and we are done.
\end{proof}

To conclude this section we note as an immediate corollary
that the Projection properties,
or the Left (Right) Residual properties, can be omitted
from the definition of bisimulation in those cases where the
corresponding operation is not primitive:

\begin{corol} \label{unnecessary}
\begin{itemize}
\item
Let $\C(\F)$ be a calculus fragment where $\F$
contains projection, and either converse is present as well or
1 is present at degree 0.  Let $\F'$ be $\F$ without 
  projection. Then $(\F,k)$-bisimilarity is the same
  as $(\F',k)$-bisimilarity.
\item
Similarly, when $\F$ contains the left (right) residual and
also converse, and 1 is present at degree 0 as well, then 
$(\F,k)$-bisimilarity is the same as
$(\F',k)$-bisimilarity where $\F'$ is $\F$ without left (right)
residual.
\end{itemize}
\end{corol}

\section{Similarity and one-sided indistinguishability}
\label{secsim}

In the present section, we deal with fragments not containing
difference, for which we will capture one-sided
indistinguishability by appropriate notions of simulation between
structures.  The treatment will largely parallel that for
fragments with difference from the previous section.
Nevertheless, simulations differ from bisimulations in that they
consist of two separate sequences $\bar Z$ and $\bar W$ of
relations, one for each direction.  The two separate directions
are needed to be able to deal with the nonmonotonic operations of
coprojection and left and right residual, in the absence of
complement and difference.

When the fragment contains neither difference, nor coprojection,
nor residuals, it will be evident that the definition of
simulation boils down to a simpler situation where only the
sequence $\bar Z$ is needed.

\subsection{General definition of $(\F,k)$-simulation}
\label{subsim}

Let $\G_1$ and $\G_2$ be two structures with node sets $V_1$ and
$V_2$ respectively.  Let $\element$ be an arbitrary element of
$V_1^2 \times V_2^2$, and let $Z$ and $W$ be arbitrary subset of $V_1^2
\times V_2^2$.

In parallel to Section~\ref{subconditions}, we define a suite of
conditions, but now appropriate for calculus fragments $\C(\F)$
containing neither complement nor difference.  The Atoms Forth
and Back are unmodified.  Also the Composition Forth and Back,
and Projection Forth and Back are unmodified with respect to
their definitions in Section~\ref{subconditions}, with the
exception that they are now defined with respect to two sets $Z$
and $W$.  The Left and Right Residual Forth and Back properties,
however, are modified in that they ``cross over'' between $Z$ and
$W$.  In the same spirit we also define Coprojection Forth and
Back properties.

\begin{description}
\item[Composition and Projection Forth]
We say that $\element$ has the \emph{Composition Forth, or
Projection Forth, property at degree $i$ with respect to $Z,W$}
if $\element$ has that property at degree $i$ with respect to
the set $Z$, as defined in Section~\ref{subconditions}.
\item[Composition and Projection Back]
We say that $\element$ has the \emph{Composition Back, or
Projection Back, property at degree $i$ with respect to $Z,W$}
if $\element$ has that property at degree $i$ with respect to
the set $W$, as defined in Section~\ref{subconditions}.
\item[Left Residual Forth]
We say that $\element$ has the
\emph{Left Residual Forth property
at degree $i$ with respect to $Z,W$} if for every $c_2$ in
$V_2$ with $(b_2,c_2)$ in $\paths^\F_{i-1}(\G_2)$, there exists
$c_1$ in $V_1$ such that $(b_1,c_1,b_2,c_2) \in W$
and either $(a_1,c_1) \notin \paths^\F_{i-1}(\G_1)$ or
$(a_1,c_1,a_2,c_2) \in Z$.

\item[Left Residual Back]
We say that $\element$ has the
\emph{Left Residual Back property
at degree $i$ with respect to $Z,W$} if
for every $c_1$ in
$V_1$ with $(b_1,c_1)$ in $\paths^\F_{i-1}(\G_1)$, there exists
$c_2$ in $V_2$ such that both $(b_1,c_1,b_2,c_2) \in Z$
and either $(a_2,c_2) \notin \paths^\F_{i-1}(\G_2)$ or
$(a_1,c_1,a_2,c_2) \in W$.

\item[Right Residual Forth]
We say that $\element$ has the
\emph{Right Residual Forth property
at degree $i$ with respect to $Z,W$} if
for every
$c_2$ in $V_2$ with $(c_2,a_2)$ in $\paths^\F_{i-1}(\G_2)$, there
exists $c_1$ in $V_1$ such that both $(c_1,a_1,c_2,a_2) \in
W$ and either $(c_1,b_1) \notin \paths^\F_{i-1}(\G_1)$ or
$(c_1,b_1,c_2,b_2) \in Z$.

\item[Right Residual Back]
We say that $\element$ has the
\emph{Right Residual Back property
at degree $i$ with respect to $Z,W$} if
for every
$c_1$ in $V_1$ with $(c_1,a_1)$ in $\paths^\F_{i-1}(\G_1)$, there
exists $c_2$ in $V_2$ such that both $(c_1,a_1,c_2,a_2) \in Z$
and either $(c_2,b_2) \notin \paths^\F_{i-1}(\G_2)$ or
$(c_1,b_1,c_2,b_2) \in W$.

\item[Coprojection Forth] We say that $\element$ has the
\emph{Projection Forth property at degree $i$ with respect to $Z,W$}
if either $a_1 \neq b_1$, or $a_1=b_1$ and $a_2=b_2$ and for
every $c_2$ in $V_2$ with $(a_2,c_2)$ in $\paths_{i-1}^\F(\G_2)$,
there exists $c_1$ in $\V_1$ such that $(a_1,c_1,a_2,c_2) \in W$.
Moreover, if $a_1=b_1$, then also for every $c_2$ in $V_2$ with
$(c_2,a_2)$ in $\paths_{i-1}^\F(\G_2)$, there must exist $c_1$ in
$\V_1$ such that $(c_1,a_1,c_2,a_2) \in W$.

\item[Coprojection Back] We say that $\element$ has the
\emph{Coprojection Back property at degree $i$ with respect to
$Z,W$} if either $a_2 \neq b_2$, or $a_2=b_2$ and $a_1=b_1$ and for
every $c_1$ in $V_1$ with $(a_1,c_1)$ in $\paths_{i-1}^\F(\G_1)$,
there exists $c_2$ in $\V_2$ such that $(a_1,c_1,a_2,c_2) \in Z$.
Moreover, if $a_2=b_2$ then also for every $c_1$ in $V_1$ with
$(c_1,a_1)$ in $\paths_{i-1}^\F(\G_1)$, there must exist $c_2$ in
$\V_2$ such that $(c_1,a_1,c_2,a_2) \in Z$.

\end{description}

Now let $k$ be a natural number and
let $\bar Z = (\Zet_0, \Zet_1, \dots, \Zet_k)$
and $\bar W = (W_0, W_1, \dots, W_k)$
be decreasing sequences of relations with
$Z_0$ and $W_0$ subsets of $V_1^2 \times V_2^2$.

We now lift the above conditions to apply to such pairs $(\bar Z,\bar
W)$ of sequences.  Plainly, the Forth properties
apply to $\bar Z$ and the Back properties to $\bar W$.
\begin{itemize}
\item
We say that $(\bar Z,\bar W)$ has the Atoms Forth property if every
element of $Z_0$ has this property.
\item
We say that $(\bar Z,\bar W)$ has the Atoms Back property if every
element of $W_0$ has this property.
\item
We say that $(\bar Z,\bar W)$ has the Composition Forth, or the
Projection Forth , or the Coprojection Forth, or the Left or
Right Residual Forth property, if for every $i\in \{1,\dots,k\}$,
every element of $Z_i$ has that property at degree $i$ with respect to
$(Z_{i-1},W_{i-1})$.
\item
We say that $(\bar Z,\bar W)$ has the Composition Back, or the
Projection Back, or the Coprojection Back, or the Left or Right
Residual Back property, if for every $i\in \{1,\dots,k\}$,
every element of $W_i$ has that property at degree $i$ with respect to
$(Z_{i-1},W_{i-1})$.
\end{itemize}

We then naturally have the following:

\begin{definition} \label{sim}
We call $(\bar Z,\bar W)$ an \emph{$(\F,k)$-simulation from $\G_1$
to $\G_2$} if $(\bar Z,\bar W)$ has
\begin{itemize}
\item
the Atoms Forth and Back properties;
\item
the Composition Forth and Back properties;
\item
the Projection Forth and Back properties when $\F$ contains
projection;
\item
the Coprojection Forth and Back properties when $\F$ contains
coprojection;
\item
the Left (Right) Residual Forth and Back properties when $\F$
contains left (right) residual.
\end{itemize}

Given two marked structures $\mG_1=(\G_1,a_1,b_1)$ and
$\mG_2=(\G_2,a_2,b_2)$, when
there exists an $(\F,k)$-simulation $(\bar Z,\bar W)$ from $\G_1$ to
$\G_2$ such that $(a_1,b_1,a_2,b_2) \in Z_k$,
we say that $\mG_1$ is \emph{$(\F,k)$-similar} to $\mG_2$,
and denote this by $\mG_1 \mysim \F k \mG_2$.
\end{definition}

It is instructive to remark that, in the case where
$\F$ does not have coprojection or residual, only $\bar Z$
matters; the component $\bar W$ is then entirely redundant,
in the sense that the sequence $\emptyset, \dots,\emptyset$
($k+1$ times) would do fine, as it trivially satisfies the
Composition and Projection Back properties.
Only when coprojection or residual are present, there is a
significant interplay between $\bar Z$ and $\bar W$.

After Definition~\ref{defbisim} of bisimilarity, we observed that
there it does not matter whether or not $\di$ belongs to the
fragment.  In contrast, here this matters, since the Atoms Forth
property applies only to $\bar Z$ and the Atoms Back property
applies only to $\bar W$.

The following important property follows immediately from the
symmetries in the definition of simulation.  For any $Z \subseteq
V_1^2 \times V_2^2$, we define $\tilde Z := \{(a',b',a,b) \mid
(a,b,a',b') \in Z\}$.  We then have:

\begin{proposition} \label{symmetric}
If $(Z_0,\dots,Z_k;W_0,\dots,W_k)$
is an $(\F,k)$-simulation from $\G_1$ to $\G_2$, then
$(\tilde W_0,\dots,\tilde W_k;\tilde Z_0,\dots,\tilde Z_k)$
is an $(\F,k)$-simulation from $\G_2$ to $\G_1$.
\end{proposition}

We also note the following analogue of Lemma~\ref{leminduction}:

\begin{lemma} \label{lemsiminduction}
Let $k>0$ and let $(Z_0,Z_1,\dots,Z_k;W_0,W_1,\dots,W_k)$ be an
$(\F,k)$-simulation from $\G_1$ to $\G_2$. Then
$(Z_0,\dots,Z_{k-1};W_0,\dots,W_{k-1})$
is an $(\F,k-1)$-simulation from $\G_1$ to $\G_2$.
\end{lemma}

\subsection{Adequacy theorem}

We establish:

\begin{theorem}[Adequacy Theorem] \label{simwater}
For any fragment $\C(\F)$ where $\F$ contains neither complement
nor difference, we have
$\mG_1 \mysim \F k \mG_2$ if and only if $\mG_1 \contndk \F k \mG_2$.
\end{theorem}

We have the following analogue of Lemma~\ref{invariant}:

\begin{lemma}[Invariance] \label{siminvariant}
If $\mG_1 \mysim \F k \mG_2$ then $\mG_1 \contndk \F k \mG_2$.
\end{lemma}
\begin{proof}
Let $e$ be an expression in $\C(\F)_k$.
We prove by induction on the structure
of $e$ that for marked structures
$\mG_1 = (\G_1, a_1,b_1) \mysim \F k \mG_2 = (\G_2, a_2, b_2)$
and any $(a_1,b_1) \in e(\G_1)$, we also have $(a_2,b_2) \in e(\G_2)$.

Let $\V_1$ and $\V_2$ be the node sets of the structures $\G_1$
and $\G_2$, respectively.  Let $(\bar Z,\bar W)$ be an
$(\F,k)$-simulation from $\G_1$ to $\G_2$ such that
$(a_1,b_1,a_2,b_2) \in Z_k$.

For the cases where $e$ is an atomic expression, a union, an
intersection, a composition, or a projection, the reasoning is
identical to the corresponding only-if cases in the proof of
Lemma~\ref{invariant}.

Consider the case where $e$ is $\cpi_1(e_1)$.  By definition of
coprojection, we have $a_1=b_1$, whence $a_2=b_2$ by the Atoms
Forth condition.  We have to show that there does not exist $c_2
\in V_2$ with $(a_2,c_2) \in e_1(\G_2)$.  For the sake of
contradiction, suppose there exists such $c_2$.  By
Lemma~\ref{leminpaths}, we have $(a_2,c_2) \in
\paths^\F_{k-1}(\G_2)$.  Then by the Coprojection Forth property,
there exists $c_1 \in V_1$ such that $(a_1,c_1,a_2,c_2) \in
W_{k-1}$.  Hence, by Proposition~\ref{symmetric},
Lemma~\ref{lemsiminduction}, and the
induction hypothesis, we obtain $(a_1,c_1) \in e_1(\G_1)$ which
is in contradiction with $(a_1,b_1) \in \cpi_1(e_1)(\G_1)$.

Finally, consider the case where $e$ is $e_1 \lres e_2$.  So we
have to show that $(a_2,b_2) \in e_1 \lres e_2(\G_2)$.
Thereto, let $c_2 \in V_2$ such that $(b_2,c_2) \in e_2(\G_2)$.
By Proposition~\ref{leminpaths}, we have that $(b_2,c_2)$ in
$\paths_{k-1}^\F(\G_2)$.  Then by the Left Residual Forth condition,
there exists $c_1$ in $\V_1$ such that $(b_1,c_1,b_2,c_2) \in W_{k-1}$
and either $(a_1,c_1) \notin \faths {k-1}(\G_1)$ or
$(a_1,c_1,a_2,c_2) \in Z_{k-1}$.
By induction, we have $(b_1,c_1) \in e_2(\G_1)$.  Hence, since
$(a_1,b_1) \in e_1 \lres e_2 (\G_1)$, we have $(a_1,c_1) \in
e_1(\G_1)$ so $(a_1,c_1) \in \faths {k-1}(\G_1)$.  Thus, the
above qualification
$(a_1,c_1) \notin \faths {k-1}(\G_1)$ is redundant, and 
$(a_1,c_1,a_2,c_2) \in Z_{k-1}$.  Again applying the induction
hypothesis we obtain $(a_2,c_2) \in e_1(\G_2)$ as desired.

The case of a right residual is completely analogous to that of a
left residual.
\end{proof}

In order to prove the simulation-analogue of the Characteristic
Expression Lemma, we now present the maximal simulation in
analogy to Definition~\ref{defmax}, and state its properties.

\begin{definition} \label{defmaxsim}
Given a fragment $\F$ as above and structures
$\G_1$ and $\G_2$ with node sets $V_1$ and $V_2$ respectively,
we construct two infinite decreasing sequences
$Z_0, Z_1, Z_2, \cdots$ and $W_0,W_1,W_2, \cdots$
by induction on $k$ as follows.
\begin{enumerate}
\item
$Z_0$ is the set of all elements of $V_1^2 \times V_2^2$ that have the
Atoms Forth property
relative to $\F$, $\G_1$ and $\G_2$.
\item
$W_0$ is the set of all elements of $V_1^2 \times V_2^2$ that have the
Atoms Back property (still relative to $\F$, $\G_1$ and $\G_2$).
\item
$Z_i$, for $i>0$, is the set of all elements in $Z_{i-1}$ that
have
\begin{itemize}
\item
the Composition Forth property
at degree $i$ with respect to $Z_{i-1}$;
\item
the Projection Forth property
at degree $i$ with respect to $Z_{i-1}$,
if $\F$ contains projection;
\item
the Coprojection Forth property at degree $i$ with respect to
$Z_{i-1},W_{i-1}$, if $\F$ contains coprojection;
\item
the Left (Right) Residual Forth property
at degree $i$ with respect to $Z_{i-1},W_{i-1}$, if $\F$ contains left
(right) residual.
\end{itemize}
\item
Similarly, $W_i$, for $i>0$, is the set of all elements in $W_{i-1}$ that
have
\begin{itemize}
\item
the Composition Back property
at degree $i$ with respect to $W_{i-1}$;
\item
the Projection Back property
at degree $i$ with respect to $W_{i-1}$,
if $\F$ contains projection;
\item
the Coprojection Back property at degree $i$ with respect to
$Z_{i-1},W_{i-1}$, if $\F$ contains coprojection;
\item
the Left (Right) Residual Back property
at degree $i$ with respect to $Z_{i-1},W_{i-1}$, if $\F$ contains left
(right) residual.
\end{itemize}
\end{enumerate}
We denote the constructed sequence $Z_0,Z_1,\dots$ by
$\Simforth \F {\G_1} {\G_2}$, and $W_0,W_1,\dots$ by
$\Simback \F {\G_1} {\G_2}$.
\end{definition}

\begin{proposition} \label{propmaxsim}

Let $\Simforth \F {\G_1} {\G_2} = Z_0, Z_1, \dots$ and $\Simback
\F {\G_1} {\G_2} = W_0, \allowbreak W_1, \dots$.  Then for each natural
number $k$, the pair of sequences $(Z_0, Z_1, \dots,
Z_k;\allowbreak W_0,W_1,\allowbreak \dots,W_k)$ is an $(\F,k)$-simulation
from $\G_1$ to $\G_2$.  Furthermore, it is the maximal simulation
in the sense that, for any other such $(\F,k)$-bisimulation
$(Z'_0, \dots, Z'_k;\allowbreak W'_0,\dots,W'_k)$, we have $Z'_i
\subseteq Z_i$ and $W'_i \subseteq W_i$ for each $i=0,\dots,k$.

\end{proposition}

\begin{corol} \label{cormaxsim}
$(\G_1,a_1,b_1) \mysim \F k
(\G_2,a_2,b_2) \; \Leftrightarrow \; (a_1,b_1,a_2,b_2) \in
\Simforth \F {\G_1} {\G_2}_k$.
\end{corol}

\begin{proposition}[Path Preservation]
\label{simpathpreserving}
Let $\Simforth \F {\G_1} {\G_2} = Z_0, Z_1, \dots$ and $\Simback
\F {\G_1} {\G_2} = W_0, W_1, \dots$, and let $i$ be a natural
number.  If $(a_1,b_1,a_2,b_2) \in Z_i$
and $(a_1,b_1) \in \paths^\F_i(\G_1)$, then
also $(a_2,b_2) \in \paths^\F_i(\G_2)$.
Similarly, if
$(a_1,b_1,a_2,b_2) \in W_i$
and $(a_2,b_2) \in \paths^\F_i(\G_2)$, then
also $(a_1,b_1) \in \paths^\F_i(\G_1)$.
\end{proposition}

We are now ready for:

\begin{lemma}[Characteristic Expression] \label{karaktersim}
Let $k$ be a natural number and let $\mG_1 = (\G_1, a_1, b_1)$ be
a marked structure.  Then there
exists an expression $e^{\F,k}_{\mG_1}$ in $\C(\F)_k$ such that
for every structure $\G_2$ we have
\[ e^{\F,k}_{\mG_1}(\G_2) = \{
(a_2,b_2) \in \paths^\F_k(\G_2) \mid \mG_1 \mysim \F k
(\G_2,a_2,b_2) \}.  \]
\end{lemma}
\begin{proof}
Analogous to the proof of Lemma~\ref{karakter}, our approach is now based
on Corollary~\ref{cormaxsim} and will
show that each of the properties $P$ involved in
Definition~\ref{defmaxsim}
is expressible in $\C(\F)_k$.  Since these properties are now
with respect to two sets $Z$ and $W$, and since the expressions
cannot use complement or difference, we must adapt the approach
from the proof of Lemma~\ref{karakter} as follows.
Let $V_1$ be the node set of the fixed
structure $\G_1$, and let $V_2$ be the node set of any structure
$\G_2$ to which our expressions will be applied.  As in the proof
of Lemma~\ref{karakter}, we 
represent the first set $Z \subseteq V_1^2 \times V_2^2$ by the family of
binary relations consisting of, for each $(a,b) \in V_1^2$, the
binary relation $$ Z_{a,b} =
\{(a',b') \in \faths {k-1}(\G_2) \mid (a,b,a',b') \in Z\}. $$
The second set $W \subseteq V_1^2 \times V_2^2$, however, is represented in a
complementary manner.  Specifically, we represent $W$ by the
family of binary relations consisting of, for each $(a,b) \in
V_1^2$, the binary relation
$$ \W_{a,b} = \{(a',b') \in \paths^\F_{k-1}(\G_2) \mid (a,b,a',b')
\notin W\}. $$
The underscore is used to remind us that the relations $\W$ give
us the complement of $W$.

As before we can add all these
relations to $\G_2$, yielding an expanded structure, denoted by
$(\G_2,Z,W)$, over the expansion of the given vocabulary $\Lambda$
with relation names $Z_{a,b}$ and $\W_{a,b}$ for each $(a,b) \in V_1^2$.
Furthermore, for our purpose, it will be sufficient to 
assume that $Z$ and $W$ are path-preserving in the sense of
Proposition~\ref{simpathpreserving}.
Now we are going to
express each property $P$ by an expression $\psigenp$
over the expanded vocabulary, in the following sense.
\begin{itemize}
\item
If $P$ is a Forth property, then for any $\G_2$ and any
$Z$ and $W$ as above, $\psigenp$ applied to $(\G_2,Z,W)$
returns the set of pairs $(a_2,b_2) \in \paths^\F_k(\G_2)$
for which $(a_1,b_1,a_2,b_2)$ satisfies the $P$ property at
degree $k$ with respect to $Z$, relative to $\G_1$, $\G_2$, and $\F$.
\item
Complementarily,
if $P$ is a Back property, then $\psigenp$ applied to
$(\G_2,Z,W)$
returns the set of pairs $(a_2,b_2) \in \paths^\F_k(\G_2)$
for which $(a_1,b_1,a_2,b_2)$ does \emph{not} satisfy the $P$
property.
\end{itemize}

It turns out that the required expressions $\psigenp$ can be
deduced easily from the expressions $\phigenp$ given in the proof
of Lemma~\ref{karakter}.  Indeed, in all these expressions,  it
turns out that whenever in some property a $Z$ must be changed to
a $W$ due to the new modified properties that cross over between
$Z$ and $W$, we already used subexpressions of the form $\faths
{k-1} - Z_{a,b}$ in the right places.  Hence, it suffices to
replace these subexpressions by $\W_{a,b}$ to obtain $\psigenp$
from $\phigenp$.  Moreover, the expressions for the Back
properties already were expressed as complements relative to
$\faths k$; there we can simply keep the parts on the right-hande
side of the set difference operator.
The correctness proofs are then completely analogous.
Specifically, the required expressions are as follows.

\begin{flushleft}
\quad
$\displaystyle
\psicompforth := \faths k \cap \bigcap_{\substack{c_1 \in \V_1\\(a_1,c_1) \in
            \paths_{k-1}^\F(\G_1)\\(c_1,b_1) \in \paths_{k-1}^\F(\G_1)}}
Z_{a_1,c_1} \circ Z_{c_1,b_1}
$
\end{flushleft}
\begin{flushleft}
\quad
$\displaystyle
\psicompback
    :=
    \paths^\F_{k} \cap
    \bigcup_{V \subseteq \V_1} \bigl (
  (\bigcap_{c_1 \in \V} \W_{a_1,c_1})
  \; \compos \;
  (\bigcap_{c_1 \in \V_1 - V} \W_{c_1,b_1}) \bigr )
$
\end{flushleft}
\begin{multline*}
\psiprojforth := \id \\
{} \cap \bigcap_{\substack{c_1 \in \V_1\\(a_1,c_1) \in
            \paths_{k-1}^{\F}(\G_1)}} \pi_1(Z_{a_1,c_1})
	    \; \cap \bigcap_{\substack{c_1 \in \V_1\\(c_1,a_1) \in
            \paths_{k-1}^{\F}(\G_1)}}
	    \pi_2(Z_{c_1,a_1})
\end{multline*}
\begin{flushleft}
\quad
$\displaystyle
\psiprojback := \id \cap
(
\pi_1 ( \bigcap_{c_1 \in \V_1} \W_{a_1,c_1})
\cup
\pi_2 ( \bigcap_{c_1 \in \V_1} \W_{c_1,a_1})
)
$
\end{flushleft}
\begin{flushleft}
\quad
$\displaystyle
\psicprojforth := 
\cpi_1(\faths {k-1} \cap \bigcap_{c_1 \in V_1} \W_{a_1,c_1})
\cap
\cpi_2(\faths {k-1} \cap \bigcap_{c_1 \in V_1} \W_{c_1,a_1})
$
\end{flushleft}
\begin{flushleft}
\quad
$\displaystyle
\psicprojback := 
\cpi_1(\bigcup_{c_1 \in V_1} Z_{a_1,c_1})
\cup
\cpi_2(\bigcup_{c_1 \in V_1} Z_{c_1,a_1})
$
\end{flushleft}
\begin{multline*}
\psilresforth
    := \faths k \\ 
    {} \cap
    \bigcap_{\V \subseteq \V_1} \Bigl [
      ( \bigcup_{c_1 \in \V_1-\V}  Z_{a_1,c_1} )
\; \lres {} \\
( \paths_{k-1}^\F
 \cap \bigcap_{\substack{c_1 \in V_1 \\ (a_1,c_1) \notin
\paths^\F_{k-1}(\G_1)}} \W_{b_1,c_1}
	\cap \bigcup_{c_1 \in V} \W_{b_1,c_1} ) \Bigr ]
\end{multline*}
\begin{multline*}
\psirresforth
    := \faths k \\
    {} \cap
    \bigcap_{\V \subseteq \V_1} \Bigl [
      ( \paths_{k-1}^\F
      \cap \bigcap_{\substack{c_1 \in V_1 \\ (c_1,b_1) \notin \faths
	{k-1}(\G_1)}} \W_{c_1,a_1}
      \cap \bigcup_{c_1 \in \V} \W_{c_1,a_1} ) \\
        {} \rres \; \bigcup_{c_1 \in \V_1-\V} Z_{c_1,b_1} \Bigr ]
\end{multline*}
\begin{flushleft}
\quad
$\displaystyle
\psilresback
    := \paths^\F_k \; \cap
        \bigcup_{\substack{c_1\in \V_1 \\ (b_1,c_1) \in
	\paths^\F_{k-1}(\G_1)}} \W_{a_1,c_1} \lres Z_{b_1,c_1} 
$
\end{flushleft}
\begin{flushleft}
\quad
$\displaystyle
\psirresback
    := \paths^\F_k \; \cap
        \bigcup_{\substack{c_1\in \V_1 \\
	(c_1,a_1) \in \paths^\F_{k-1}(\G_1)}}
	Z_{c_1,a_1} \rres \W_{c_1,b_1}
$
\end{flushleft}

We are now ready to present the construction of 
the required expression $e^{\F,k}_{\mG_1}$, by induction on $k$.
Actually, we will simultaneously construct an expression
$e'^{\F,k}_{\mG_1}$ with the property that for every $\G_2$ we
have
\[ e'^{\F,k}_{\mG_1}(\G_2) = \{
(a_2,b_2) \in \paths^\F_k(\G_2) \mid (a_1,b_1,a_2,b_2) \notin
\Simback \F {\G_1} {\G_2}_k \}, \]
thus complementing expression 
$e^{\F,k}_{\mG_1}$ which must satisfy
\[ e^{\F,k}_{\mG_1}(\G_2) = \{
(a_2,b_2) \in \paths^\F_k(\G_2) \mid
(a_1,b_1,a_2,b_2) \in \Simforth \F {\G_1} {\G_2}_k \}. \]

For the base of the construction, we put
\begin{align*}
& e^{\F,0}_{\mG_1} := \faths 0 \cap \bigcap_{e \in \atype \F{\mG_1}} e ; \\
& e'^{\F,0}_{\mG_1} := \faths 0 \cap \bigcup_{e \notin \atype \F{\mG_1}} e .
\end{align*}

For $k>0$, expression $e^{\F,k}_{\mG_1}$, applied to any $\G_2$,
should return the set of all pairs $(a_2,b_2) \in \faths k(\G_2)$
such that $(a_1,b_1,a_2,b_2) \in \Simforth \F {\G_1} {\G_2}_k$.
By definition, these are the pairs $(a_2,b_2)$ for which
$(a_1,b_1,a_2,b_2) \in \Simforth \F {\G_1} {\G_2}_{k-1}$, and
such that $(a_1,b_1,a_2,b_2)$ has all the Forth properties
required by Definition~\ref{defmaxsim}, at degree $k$ with
respect to $\Simforth \F {\G_1} {\G_2}_{k-1}, \Simback \F {\G_1}
{\G_2}_{k-1}$.  We have just seen that each Forth property $P$ is
expressible by $\psigenp$.  Hence, we can obtain
$e^{\F,k}_{\mG_1}$ simply as the intersection of
$e^{\F,k-1}_{\mG_1}$ (obtained by induction) and the expressions
$\psigenp$ for the different Forth properties $P$ required by
Definition~\ref{defmaxsim}.

Complementarily, expression $e'^{\F,k}_{\mG_1}$ should return the
set of all pairs $(a_2,b_2) \in \faths k(\G_2)$ such that
$(a_1,b_1,a_2,b_2) \notin \Simback \F {\G_1} {\G_2}_k$.  This
means that $(a_1,b_1,\allowbreak a_2,b_2)$ must \emph{not} satisfy at least
one of the Back properties required by
Definition~\ref{defmaxsim}, at degree $k$ with respect to
$\Simforth \F {\G_1} {\G_2}_{k-1}, \Simback \F {\G_1}
{\G_2}_{k-1}$.  We have just seen that the complement of each
Back property $P$ is expressible by $\psigenp$.  Hence, we can
obtain $e^{\F,k}_{\mG_1}$ simply as the intersection of $\faths
k$ with the union of the expressions $\psigenp$ for the different
Back properties $P$ required by Definition~\ref{defmaxsim}.

In analogy to the proof of Lemma~\ref{karakter}, in the resulting
expressions, we replace each relation name $Z_{a,b}$ by 
$e^{\F,k-1}_{\G_1,a,b}$; furthermore, we replace each relation
name $\W_{a,b}$ by $e'^{\F,k-1}_{\G_1,a,b}$.
The reduction to a finitary expression, based on quantifier
rank, is exactly as in the proof of Lemma~\ref{karakter}, and we
are done.
\end{proof}

\section{Indistinguishability of finite structures}

The bisimilarity characterizations we have given of when two
structures are indistinguishable by expressions of $\C(\F)_k$,
for some fixed degree $k$ and some fixed fragment $\F$, are valid
for arbitrary structures.  For finite structures, by classical
arguments \cite{hennessymilner,gorankotto,blackburn_modallogic},
our methods lead immediately to Hennessy--Milner-style theorems
about indistinguishability in the full fragment $\C(\F)$, without
a degree restriction, as we will show in the present section.  It
also follows that indistinguishability in $\C(\F)$ is decidable
in polynomial time.

\subsection{Bisimulation without degree restriction}

Let $\C(\F)$ be a fragment containing complement or difference.
We want to define a natural notion of $\F$-bisimulation without a
degree restriction $k$.
Thereto we must make two small adaptations.
\begin{enumerate}
\item
Recalling Definition~\ref{defgraph}, 
let us define $\fpaths(\G)$ as the set of all pairs $(x,y)$ in
$V^2$ such that
\begin{itemize}
\item
there is a directed path from $x$ to $y$ in $\graph \G$, if $\F$
does not contain converse; or
\item
there is an undirected path from $x$ to $y$ in $\ugraph \G$, if $\F$
contains converse.
\end{itemize}
\item
Recall the suite of Forth and Back properties
introduced in Section~\ref{subconditions}.  We naturally
introduce variants of these properties that are no longer
degree-restricted.  It suffices to replace each reference to
$\faths i$ or $\faths {i-1}$
by $\fpaths$, so that the degree parameter $i$ becomes
irrelevant.
\end{enumerate}

We now define an \emph{$\F$-bisimulation from $\G_1$ to $\G_2$}
as a relation $Z \subseteq V_1^2 \times V_2^2$ such that each of
its elements has the Atoms Forth and Back properties, as well as
the (degree-unrestricted versions of the) Composition Forth and
Back properties with respect to $Z$, and also the Projection
(Left Residual, Right Residual) Forth and Back properties with
respect to $Z$ depending on whether $\F$ contains projection
(left residual, right residual), as usual.  When there exists an
$\F$-bisimulation from $\G_1$ to $\G_2$ containing
$(a_1,b_1,a_2,b_2)$, we say that $\mG_1=(\G_1,a_1,b_1)$ and
$\mG_2=(\G_2,a_2,b_2)$ are \emph{$\F$-bisimilar} and denote this
by $\G_1 \simeq^\F \G_2$.

We show:

\begin{theorem} \label{hennessymilner}
For finite marked structures $\mG_1$ and $\mG_2$,
we have $\mG_1 \simeq^\F \mG_2$ if
and only if $\mG_1 \indist \F \mG_2$.
(The only-if implication actually holds for all structures,
finite or infinite.)
\end{theorem}
\begin{proof}
Let $\mG_1=(\G_1,a_1,b_1)$ and $\mG_2=(\G_2,a_2,b_2)$, and
assume $\mG_1 \simeq^\F \mG_2$ by the bisimulation $Z$.
In order to show $\mG_1 \indist \F \mG_2$, we must show that
$\mG_1 \indistk \F k \mG_2$ holds
for all degrees $k$.  So, let $k$ be arbitrary.
Consider the sequence $\bar Z=Z,Z,\dots,Z$ that simply consists of $k+1$
times $Z$.  It is now readily verified that $\bar Z$ is an
$(\F,k)$-bisimulation from $\G_1$ to $\G_2$.  Since
$(a_1,b_1,a_2,b_2) \in Z$, we conclude $\mG_1 \indistk \F k \mG_2$
by Lemma~\ref{invariant}.

Conversely, assume $\mG_1 \indistk \F k \mG_2$ for every degree $k$.
This means that $(a_1,b_1,a_2,\allowbreak
b_2) \allowbreak \in \BiSim \F {\G_1} {\G_2}_k$
for every $k$.  Let $\BiSim \F {\G_1} {\G_2} = Z_0,Z_1,\dots$.
Recall that this is a decreasing sequence.  Hence, since $\G_1$
and $\G_2$ are finite, there exists a degree $p$ such that
$Z_{q+1}=Z_q$ for all $q \geq p$.  As a consequence,
$Z_p$ is an $\F$-bisimulation from $\G_1$ to $\G_2$ and we are
done.  
\end{proof}

\subsection{Simulation without degree restriction}

For any fragment $\C(\F)$ containing neither complement nor difference,
in an entirely analogous manner
we can define degree-unrestricted versions of the properties 
for simulations listed in Section~\ref{subsim}.  Then an
{$\F$-simulation from $\G_1$ to $\G_2$} is a pair $(Z,W)$ of
subsets of $V_1^2 \times V_2^2$ such that each element of $Z$
has all Forth properties with respect to $(Z,W)$ corresponding
to the operators present in $\F$, and each element of $W$ has all
Back properties with respect to $(Z,W)$.
We can then again show that $\mG_1 \contnd \F \mG_2$ if and only
if there exists an $\F$-simulation $(Z,W)$ from $\G_1$ to $\G_2$
such that $(a_1,b_1,a_2,b_2) \in Z$.

\subsection{Polynomial-time complexity}

As a corollary of the above, we obtain:

\begin{corol}
For any fixed fragment $\C(\F)$, it can be decided in polynomial time
whether or not two given finite marked structures are
indistinguishable in $\C(\F)$.
\end{corol}
\begin{proof}
We give the proof for fragments with complement or difference;
the case for the other fragments is entirely analogous.
Let $\mG_1=(\G_1,a_1,b_1)$ and $\mG_2=(\G_2,a_2,b_2)$
be finite marked structures.  Let $n_1$ ($n_2$) be the number of
nodes of $\G_1$ ($\G_2$).  Let $\BiSim \F {\G_1} {\G_2} = Z_0,Z_1,\dots$.
Recall from the proof of Theorem~\ref{hennessymilner} that
there exists $p$ such that $Z_{q+1}=Z_q$ for all $q\geq p$, and
$\mG_1 \indist \F \mG_2$ if and only if $(a_1,b_1,a_2,b_2) \in
Z_p$.  Let $p$ be the smallest such $p$.  Then in the worst case, for
each $i<p$, there is exactly one element less in $Z_{i+1}$
compared to $Z_i$.  Hence, $p \leq n_1^2n_2^2$ which is
polynomial in $n_1$ and $n_2$.  Moreover, it is evident from
Definition~\ref{defmax} that $Z_{i+1}$ can be computed from $Z_i$
in time polynomial in $n_1$, $n_2$, and the size of $Z_i$ which
is itself bounded by $n_1^2n_2^2$.  Hence, we can
compute $Z_p$ by performing a polynomial number of iterations
where each iteration takes polynomial time, and we are done.
\end{proof}

\section{Concluding remarks}
\label{sec:conclusion}

We have always included the identity relation and the three
operations union, intersection and composition in the logics that
we consider.  As already mentioned in the Introduction, it is an
interesting topic for further research to see what happens if
some of these operators are left out.

The results of the present paper provide the tools to continue
the research on the relative expressive power of fragments of the
calculus of relations.  In earlier work \cite{rafragments} the
precise relationships between all fragments were clarified,
ignoring residuals however.  Since in the present paper we have
fully integrated the residuals, an interesting direction for further
research is now to throw the residuals in the picture.

For example, it is an intriguing question how the
fragments $\C(-)$ and $\C(\lres,\rres)$ relate to each other 
in their power to express boolean queries (where `true' is
represented by any nonempty answer and `false' is represented by
the empty answer).  Both fragments extend the basic fragment $\C$
with significant nonmonotonic operators, viz., difference on the
one hand and the residuals on the other hand.
We actually conjecture that the two expressive
powers are incomparable.  For example, following the standard approach
\cite{EF99,Libkin04,ahv_book}, one may try to prove that
$\C(\lres,\rres)$ is not subsumed by $\C(-)$ for boolean queries
by exhibiting an expression $e$ in $\C(\lres,\rres)$ and,
for each degree $k$, two structures $\A_k$ and
$\B_k$ such that for each $k$ the following holds:
\begin{itemize}
\item
$e(\A_k) \neq \emptyset = e(\B_k)$;
\item
for any pair $(a,b)$ of nodes of $\A_k$ there exists a pair
$(a',b')$ of nodes of $\B_k$ such that
$(\A_k,a,b) \mysim {\{\lres,\rres\}} k (\B_k,a',b')$.
\end{itemize}
This indeed implies that no expression $e'$ in
$\C(\lres,\rres)$ can correctly express the boolean query
expressed by $e$.  For, because $e(\A_k) \neq \emptyset$,
any such expression should return at least one pair $(a,b)$
on $\A_k$, but then by the similarity 
relationship, it will also return a pair $(a',b')$ on $\B_k$, in
contradiction with $e(\B_k)=\emptyset$.

For fragments with set difference or complementation, we have
worked with bisimulations; for fragments without set difference,
we have worked with simulations.  In modal logic there are
results \cite[Theorem 2.78]{blackburn_modallogic} that turn this
around, showing that formulas invariant under simulations are in
fact equivalent to positive-existential formulas.  It is an
interesting direction to see how such a result could be
formulated in the setting considered in the present paper.

\section*{Acknowledgment}

We thank the two anonymous referees for their critical comments
on an earlier draft of this paper.


\end{document}